\documentclass[10pt]{article}
\usepackage{authblk}
\usepackage{amsfonts}
\usepackage{float}
\usepackage{hyperref}
\usepackage{caption}
\usepackage{multirow}
\usepackage{subcaption}

\usepackage[switch, modulo]{lineno}

\usepackage{amssymb,latexsym,amsmath,amscd,slashed,mathtools,mathrsfs,bm,amsthm}
\usepackage{setspace,scalerel,color,graphicx,enumerate}

\newcommand{\res}{\mathrm{Res}}

\theoremstyle{plain}
\newtheorem{theorem}{Theorem}
\newtheorem{corollary}[theorem]{Corollary}

\newtheorem{prop}[theorem]{Proposition}

\theoremstyle{definition}

\DeclareMathOperator{\tr}{Tr}
\DeclareMathOperator{\ex}{\mathbb{E}}

\theoremstyle{plain}

\makeatletter

\date{September 17, 2026}
\makeatother

\title{Asymmetric phase transitions in random noncommutative geometries}

\author{
Benedek Bukor$^{1}$,
Masoud Khalkhali$^{2}$,
Samuel Kov\'{a}\v{c}ik$^{1}$,
Adam Kubiš$^{1}$,
Katarína Magdolenová$^{1}$,
Nathan Pagliaroli$^{3}$,
Juraj Tekel$^{1}$ \thanks{Email: \emph{Email addresses}:benedek.bukor@fmph.uniba.sk, masoud@uwo.ca, samuel.kovacik@fmph.uniba.sk, 
\newline kubis31@uniba.sk, magdolenova14@uniba.sk, npagliar@uwaterloo.ca, juraj.tekel@fmph.uniba.sk}
\\
$^{1}$Faculty of Mathematics, Physics and Informatics,
Comenius University in Bratislava Bratislava, Slovakia\\
$^{1}$Department of Mathematics, The University of Western Ontario, London, ON, Canada\\
$^{3}$Department of Combinatorics and Optimization, The University of Waterloo, Kitchener, ON, Canada} 

\begin{document}
	\maketitle
	\begin{abstract}
 In this paper, we study the asymmetric phases of the quartic type $(0,1)$ and $(1,0)$ Dirac ensembles via three approaches: the Riemann-Hilbert approach, bootstrapping with positivity, and Hamiltonian Monte Carlo (HMC) simulations. The focus of this work is on asymmetric solutions to the Schwinger-Dyson and saddle point equations of these models, whose solution spaces prove deeply intricate. Via the Riemann-Hilbert approach, we are able to give explicit formulae for the eigenvalue distributions and free energy of various solutions. Using Hamiltonian Monte Carlo simulations, we are able to reconstruct the phase structure. Lastly, using bootstrapping with positivity, we are able to reconstruct the eigenvalue distribution of these models from their bootstrapped moments. All three methods show excellent agreement for a large matrix size.
	\end{abstract}
	\tableofcontents
	\section{Introduction}
In noncommutative geometry, spectral triples are noncommutative analogues of Riemannian manifolds. Fundamentally, the triple consists of a complex involutive algebra that acts by bounded operators on a Hilbert space, with a self-adjoint Dirac operator $D$ \cite{connes1995noncommutative}. Motivated by this paradigm, in \cite{barrett2015matrix}, Barrett proposed toy models of quantum gravity that integrate over spaces of finite spectral triples. These particular finite spectral triples have matrix algebras which allow for well-defined partition functions of such models as matrix integrals. More precisely, a \textit{fuzzy spectral triple} of signature $(p, q)$ is a real finite spectral triple of
the form
$$(M_{N}(\mathbb{C}), V\otimes M_{N} (\mathbb{C}), D, J, \Gamma ),$$
where $V$ is some irreducible Clifford module for the Clifford algebra $Cl_{p,q}$ with real structure $J$ and grading $\Gamma$, and $D$ is subject to the axioms of real spectral triples \cite{barrett2015matrix}. 

A \textit{Dirac ensemble} is a set of fuzzy spectral triples with all data fixed except for the Dirac operator, for which a probability distribution is given. Motivated by the spectral action principle \cite{connes1996gravity}, we assign a probability distribution of the form
\begin{equation*}
 \frac{1}{Z} e^{-S(D)}dD.
\end{equation*}
Due to the classification theorem of Barrett\cite{barrett2015matrix}, we know that for a fixed signature $(p,q)$ of Clifford algebra, the Dirac operator of a fuzzy geometry can be written as a certain finite sum of anti-commutators and commutators of Hermitian matrices tensored to products of gamma matrices. Hence, Dirac ensembles are in fact a rather specific class of random Hermitian matrices. Recent work has aimed to expand the framework of these models by introducing Yang-Mills-Higgs contributions \cite{perez2022multimatrix} and fermions \cite{barrett2024fermion,khalkhali2025large}.

 In this work, we restrict our attention to the type $(1,0)$ and $(0,1)$ fuzzy geometries whose Dirac operators are of the form 
		$$D_{(1,0)} = \{H,\cdot\} \quad \text{and}\quad D_{(0,1)} = [H,\cdot],$$
		respectively, for a Hermitian $N$ by $N$ matrix $H$.
For these signatures, we will study quartic Dirac ensembles with $S(D)=t_{2}\tr D^2 + t_{4}\tr D^{4}$. In terms of $H$, the probability distribution can be written as 
\begin{equation} \label{action}
 \frac{1}{Z} \exp\left(-2t_{2}( N\tr H^2 +\epsilon(\tr H)^2) -2t_{4}(N \tr H^{4} +4 \epsilon \tr H \tr H^3 + 3(\tr H^2)^2)\right)dH,
\end{equation}
where $\epsilon =1$ for type $(1,0)$ and $\epsilon =-1$ for type $(0,1)$. The measure $dD=dH$ is the Lebesgue measure on the space of $N \times N$ Hermitian matrices:
		\begin{equation*}
		d D = d H = \prod_{i=1}^N d H_{ii} \, \prod_{1 \leq i < j \leq N} d ({\text{Re}} (H_{ij})) \, d ({\text{Im}} (H_{ij})). 
		\end{equation*}
In the case of the type $(0,1)$ geometry, there are some additional nuances to defining the integral due to the non-trivial kernel of the commutator. The mapping from $D$ from $H$ is therefore not injective, and therefore there is not a pullback measure from $dD$. To make this measure well-defined to each fiber we equip the fiber $\left\{H=H_0+t I\right\}_{t \in \mathbb{R}}$ with the probability density $\sqrt{\frac{a}{\pi}} e^{-a \operatorname{Tr}(H)^2}$ for a constant $a>0$.

These models have been studied in previous works both analytically in the large $N$ limit \cite{khalkhali2020phase,khalkhali2022spectral,hessam2023double} and numerically \cite{barrett2016monte,d2022numerical,perez2021multimatrix,glaser2017scaling}. Most prior analytic works have studied only symmetric solutions due to the symmetry of the associated matrix integrals. This paper is dedicated to studying asymmetric solutions of these models, which have vastly richer solution spaces and phase structures. Asymmetric solutions of multi-tracial matrix models also arise in the context of fuzzy field theories \cite{tekel2015matrix,tekel2018asymmetric,Bukor:2024kqy,prekrat2023approximate}. It would be interesting to find deeper connections between Dirac ensembles and fuzzy field theories. 

In this work, we will study these quartic Dirac ensembles through three approaches. First, in the large-$N$ limit, the distribution of eigenvalues is known to minimise a functional equation that can be phrased as a scalar Riemann-Hilbert problem. Solving the resulting saddle point equation gives the distribution of eigenvalues for large $N$. From the saddle-point equation, we will derive explicit algebraic expressions for the density function and free energy. The recent thesis and paper \cite{d2022numerical,d2026symmetry} also employ these same techniques to study asymmetric solutions to the saddle point equations of these models.

Second, we will study these models using the bootstrapping with positivity method \cite{lin2020bootstraps,kazakov2022analytic,li2025analytic,khalkhali2025bootstrapping}. The moments of matrix models satisfy an infinite system of recursive equations known as the Schwinger-Dyson equations. The moments of matrix models are also subject to positivity constraints related to the Hamburger moment problem. Combining the Schwinger-Dyson equations and these positivity constraints into a non-linear optimisation problem is known as bootstrapping with positivity. Dirac ensembles have previously been bootstrapped in \cite{hessam2022bootstrapping}; however, the solution space was restricted to only symmetric solutions for reasons discussed below. We are then able to reconstruct the eigenvalue distributions of these models.

Third, we use Hamiltonian Monte Carlo (HMC) simulations. We are studying a single-matrix multi-trace model and such models are usually well-handled by this method. However, typically studied simple pure potential or fuzzy space models often have a vacuum structure that is symmetric in some respects; for example, the eigenvalues move close to one of the potential minima, $\pm \Phi_0$. In this model, we are dealing with a more complex structure in which solutions may be asymmetric, and multiple solutions of different structures coexist or have similar free energies. One can overcome this issue by initialising the simulation close to the correct a priori known solution, as done in \cite{d2026symmetry}. Instead, we employ a new algorithm described in \cite{Kovacik:2026ybz}.

All three methods show excellent agreement with each other as well as with the independent study of these models in the recent article \cite{d2026symmetry}. In the quartic type $(0,1)$ model, there is one phase transition between a symmetric one-cut and a symmetric two-cut solution. In the quartic type $(1,0)$ model, there are two phase transitions between an asymmetric two-cut solution and the symmetric one-cut solution.

Note that the quartic type $(0,1)$ and $(1,0)$ Dirac ensembles are symmetric in their matrix variable, hence all odd tracial moments of the matrix integrals \eqref{action} must be zero. Thus, asymmetric solutions of the saddle point equation or Schwinger-Dyson equations do not correspond to a matrix integral in either the formal or convergent interpretation of the model. The desire to find solutions corresponding to matrix integrals is why the authors of the work \cite{khalkhali2020phase} restricted their attention to solely symmetric solutions. Instead, asymmetric solutions as mathematical objects are the solutions of the saddle point equation or Schwinger-Dyson equations. In essence, this paper and \cite{d2026symmetry} are studying such solutions, which are found to be much richer in features than if one considered only symmetric solutions. 

In Section \ref{sec:RH}, we study the various solutions to the saddle point equations of both models. We derive algebraic constraints on the solution spaces and explicit formulae for the eigenvalue densities and free energies of various solutions. Section \ref{sec:bootstraps} combines the Schwinger-Dyson equations and positivity constraints to estimate the moments and eigenvalue distributions of the models. In Section \ref{sec:MC}, HMC simulations are used to study the free energy and eigenvalue distributions of these models. In Section \ref{sec:BS}, we compare the various approaches. Section \ref{sec:conclusion} summarizes our results and discusses future work.
In Appendix A and B, respectively, we derive the constraints on one-cut solutions and two-cut solutions. Appendix C details the computation of the free energy for one-cut solutions.

\section{The Riemann-Hilbert approach}\label{sec:RH}
	As studied in \cite{khalkhali2025large,khalkhali2020phase,d2022numerical,d2026symmetry}, the equilibrium measure $\mu$ minimizes the free energy functional
 \begin{equation*}
 I[\mu] = \int \int\left(U(x,y) - \ln|x-y|\right)d\mu(x)d\mu(y), 
 \end{equation*}
 where 
 \begin{align*}
 U(x,y) &=t_{4}(2(x^4 + y^4) + 12 x^2 y^2 + 8 \epsilon(x^3 y + xy^3)) + t_{2}(2(x^2+y^2)+4 \epsilon xy) + \delta_{(0,1)} 2a x y
 \end{align*}
 and ${\epsilon =1}$ for type $(1,0)$ and ${\epsilon =-1}$ for type $(0,1)$. Note that the free energy of such models in the limit is equal to the free energy functional evaluated at its minimum, i.e.
 \begin{equation*}
 \lim_{N \rightarrow \infty}\frac{1}{N^2}\ln Z = -I[\mu_{\text{min}}].
 \end{equation*}
 
We are interested in minimizing measures with compactly supported continuous density functions, i.e. $\mu = \rho(x) dx$ where $\text{supp}(\rho) =\bigcup_{i=1}^{k}[a_{i},b_{i}]$. Assuming the solution is of this form and by taking the weak derivative of $I$, one can show that any such minimizing $\mu$ must satisfy the saddle point equation
\begin{equation}
 \text{P.V.} \int\limits_{\text{supp}(\rho)}\frac{\rho(y)}{x-y}dy = \frac{1}{2}\frac{d}{dx}\int U(x,y)\rho(y)dy = 4 t_{4} x^3 +12 t_{4} m_{2}x + 12 \epsilon t_{4}m_{1}x^2 + 4 \epsilon t_{4} m_{3} + 2t_{2} x +2 \epsilon t_{2} m_{1} + \delta_{(0,1)} m_{1}a ,
\end{equation}
where P.V. denotes the Cauchy principal value. Denote the right-hand side as $Q$. This problem can be stated and solved as a Riemann-Hilbert problem. Define 
\begin{equation*}
 G_{\pm}(z) =\lim_{\varepsilon\rightarrow 0^{+}} \frac{1}{i \pi\sqrt{q(z)}}\int \frac{\rho(y)}{y -(z\pm i \varepsilon)}dy
\end{equation*}
on $\mathbb{C}\setminus \text{supp}(\rho)$ where 
\begin{equation*}
 q(z) = \prod_{i=1}^{k}(z-a_{i})(z-b_{i}).
\end{equation*}
The Riemann-Hilbert problem can then be succinctly written as
\begin{equation*}
 G_{+}(z) - G_{-}(z) = \begin{cases} 
 \frac{1}{2 \pi i}\int\limits_{\text{supp}(\rho)} \frac{\frac{i}{\pi}Q(s)}{(s-z)(\sqrt{q(s)})_{+}}ds & z \in\text{supp}(\rho) \\
 0 & z\in \mathbb{R}\setminus \text{supp}(\rho).
 \end{cases}
\end{equation*}

 This is a scalar Riemann-Hilbert problem that can be solved using the Sokhotski–Plemelj formula: 
 \begin{equation*}
 G(z) :=\frac{1}{i \pi} \int\frac{\rho(y)}{y-z}dy = \frac{\sqrt{q(z)}}{2 \pi i}\int \limits_{\text{supp}(\rho)} \frac{\frac{i}{\pi}Q(s)}{(s-z)(\sqrt{q(z)})_{+}}ds.
 \end{equation*}
The complex function $G$ has $2k$ constraints, which can be derived as follows. Consider 
\begin{equation*}
 G(z) = -\frac{(z^{k} + ...)}{2 \pi iz}\int\limits_{\text{supp}(\rho)} \frac{\frac{i}{\pi}Q(s)}{(\sqrt{q(z)})_{+}}\left(1+\frac{s}{z} +...+\left(\frac{s}{z}\right)^{k-1} +... \right)ds.
\end{equation*}
To ensure that $G(z)$ goes to zero as $z$ goes to infinity, it must be the case that 
\begin{equation}
 \int\limits_{\text{supp}(\rho)} \frac{\frac{i}{\pi}Q(s)}{(\sqrt{q(z)})_{+}}s^{j}ds =0
\end{equation}
for $0\leq j \leq k-1$.
For $j=k$, 
\begin{equation}
 \frac{i}{2 \pi}\int\limits_{\text{supp}(\rho)} \frac{Q(s)}{(\sqrt{q(z)})_{+}}s^{k}ds =1,
\end{equation}
since $G(z) = -\frac{1}{i \pi z} + \mathcal{O}(z^{-2})$.
Additionally, one can show that 
\begin{equation}
 \int\limits_{b_{j}}^{a_{j+1}}\left(H [\rho(x)] - \frac{1}{2\pi } Q(x) \right) dx = 0 
\end{equation}
for $1 \leq j \leq k-1$, where we denoted 
\begin{equation*}
 H [\rho(x)]:= \frac{1}{\pi}\text{P.V.} \int\frac{\rho(y)}{x-y}dy
\end{equation*}
as the Hilbert transform. We can simplify the integral of the Hilbert transform as 
\begin{align*}
 \int\limits_{b_{j}}^{a_{j+1}} H[\rho(x)]dx &=\frac{1}{\pi}\int\limits_{b_{j}}^{a_{j+1}} \int\frac{\rho(y)}{x-y}dy dx\\
 &=\frac{1}{\pi} \int \ln|a_{j+1}-y|\rho(y) dy - \frac{1}{\pi}\int \ln|b_{j}-y|\rho(y) dy.
\end{align*}

In the case of single tracial matrix ensembles, this is enough information to explicitly determine the support and moments in terms of the coupling constants. However, for multi-tracial integrals such as ours, we require some constraints on the moments. Suppose that such a potential's saddle point equation contains the moments $\{m_{i_{1}},m_{i_{2}},...,m_{i_{p}}\}$, then such constraints can be found by taking the ansatz for $\rho$ in terms of the couplings and these moments, and then considering the non-linear system of equations given by 
$$m_{i_{j}} = \int\limits_{\text{supp}(\rho)}x^{i_{j}}\rho(x)dx,$$ 
given by evaluating the right-hand side for the ansatz of $\rho$ depending on the number of cuts of the solution. In the following subsections, we will simplify these constraints and derive algebraic formulae for the density function of the equilibrium measure.

\subsection{One-cut solutions}
	Assume the support is $\text{supp} (\rho) = [a,b]$, with $a\not =b$. An analysis can be carried out using the Zhukovsky transform in more generality than in \cite{khalkhali2020phase} and similarly to that which was done in \cite{d2026symmetry}, in order to find equilibrium measures. We may summarise the above possible forms of the continuous density function in the following theorem. For the proof, see Appendix \ref{AppA}.
 
	\begin{theorem}\label{thm: one-cut density}
			The spectral density function of the one-cut solution of the type $(1,0)$  quartic Dirac ensemble is
		\begin{equation*}
		\rho(x) = \frac{1}{2\pi} \left(\frac{1}{\gamma ^2}-\frac{24 \epsilon t_{4} (x-\alpha ) \left(\alpha +24 \alpha \gamma ^4 t_4\right)}{24 \gamma ^4
 \epsilon t_{4}-1}+8 t_{4} \left(-2 \alpha ^2-\gamma ^2+x^2+\alpha x\right)\right)\sqrt{(x-a)(b-x)}_{[a,b]},
		\end{equation*}
 where $\alpha = \frac{a+b}{2}$ and $\gamma =\frac{b-a}{4}$ are the roots of the polynomials \eqref{eq:1-cut condition 2} and \eqref{eq:1-cut condition 1} and $\sqrt{\ldots}_\mathcal I$ denotes a function vanishing outside of the set of intervals $\mathcal I$.
	\end{theorem}

 Note that when the solutions are symmetric, the odd moments are zero and all terms containing the sign $\epsilon$ are zero. Hence, any symmetric one-cut solution is the same for both models. 

\subsubsection{Symmetric one-cut solution}
Assuming that the support of $\rho$ is symmetric and one interval, it was first found in \cite{khalkhali2020phase,d2022numerical} that
\begin{align*}
\rho(x)&= \frac{1}{2\pi}\left(8 t_{4}(x^2-\gamma^2) +\frac{1}{\gamma^2}\right)\sqrt{4\gamma^{2}-x^{2}}_{[-2\gamma,2\gamma]},
\end{align*}
where the support $[-2\gamma,2\gamma]$ is the real positive root of
\begin{align} \label{eq:1-cut tree equation}
192 t_{4}^2\gamma^8 + 48t_{4} \gamma^4 + 4t_{2} \gamma^2 = 1,
\end{align}
which can be written as 
\begin{equation*}
 \gamma^2 = \frac{1}{4 \sqrt{3}}\left(\sigma \sqrt{ -\frac{2}{t_{4}}+R}+\sqrt{ -\frac{4}{t_{4}}-\sigma \frac{\sqrt{3}t_{2}}{t_{4}^{2} \sqrt{R}} -R}\right),
\end{equation*}
where 
\begin{equation*}
 R :=\sqrt[3]{\frac{3 t_{2}^2 + 32 t_{4}}{4 t_{4}^{4}}}
\end{equation*}
and $\sigma = 1$ for $t_{2}\leq 0$ and $\sigma =-1$ for $t_{2}>0$. 

\subsection{Two-cut solutions}\label{sec:2cut solutions}
Assuming $\text{supp}(\rho) = [a_{1},b_{1}]\cup[a_{2},b_{2}]$, in \cite{d2022numerical}, it was shown that the density function is of the form
\begin{equation*}
 \rho(x) = \frac{1}{2 \pi} \left(24 \epsilon m_{1} t_4 +8 t_4 x + 4t_4(a_{1}+a_{2}+b_{1}+b_{2})\right)\text{cut}(x)\sqrt{-(x-a_{1})(x-a_{2})(x-b_{1})(x-b_{2})}_{[a_{1},b_{1}]\cup[a_{2},b_{2}]},
\end{equation*}
where 
\begin{equation*}
 \operatorname{cut}(x)= \begin{cases}-1 & x \in\left[a_1, b_1\right] \\ +1 & x \in\left[a_2, b_2\right]\end{cases}.
\end{equation*}
From the discussion at the beginning of this section, we know that we have the seven constraints. The algebraic solutions to these constraints are rather complicated but provide tameable formulae, though not so simple, as can be seen below. For example, $m_{1}$ can be written in terms of the 
limit points 
as 
\begin{equation*}
 \resizebox{1.05\hsize}{!}{$m_{1}=-\frac{2 \left(-3 a_{1}^5+a_{1}^4 (a_{2}+b_{1}+b_{2})+2 a_{1}^3
 \left(a_{2}^2+b_{1}^2+b_{2}^2\right)+2 a_{1}^2 \left(a_{2}^3-a_{2}^2
 (b_{1}+b_{2})-a_{2} \left(b_{1}^2+b_{2}^2\right)+(b_{1}-b_{2})^2
 (b_{1}+b_{2})\right)+a_{1} \left(a_{2}^4-2 a_{2}^2
 \left(b_{1}^2+b_{2}^2\right)+\left(b_{1}^2-b_{2}^2\right)^2\right)-3
 a_{2}^5+a_{2}^4 (b_{1}+b_{2})+2 a_{2}^3 \left(b_{1}^2+b_{2}^2\right)+2
 a_{2}^2 (b_{1}-b_{2})^2 (b_{1}+b_{2})+a_{2}
 \left(b_{1}^2-b_{2}^2\right)^2-(b_{1}-b_{2})^2 \left(3 b_{1}^3+5 b_{1}^2
 b_{2}+5 b_{1} b_{2}^2+3 b_{2}^3\right)\right)}{-15 a_{1}^4+12 a_{1}^3
 (a_{2}+b_{1}+b_{2})+6 a_{1}^2 \left(a_{2}^2-2 a_{2}
 (b_{1}+b_{2})+(b_{1}-b_{2})^2\right)+12 a_{1} \left(a_{2}^3-a_{2}^2
 (b_{1}+b_{2})-a_{2} (b_{1}-b_{2})^2+(b_{1}-b_{2})^2
 (b_{1}+b_{2})\right)-15 a_{2}^4+12 a_{2}^3 (b_{1}+b_{2})+6 a_{2}^2
 (b_{1}-b_{2})^2+12 a_{2} (b_{1}-b_{2})^2 (b_{1}+b_{2})-15
 b_{1}^4+12 b_{1}^3 b_{2}+6 b_{1}^2 b_{2}^2+12 b_{1} b_{2}^3-15
 b_{2}^4+32}$,}
\end{equation*}
and $m_{2}$ can be expressed in terms of the limit points and $m_{1}$ as
\begin{align}\label{eq:m2 2-cut}
\begin{split}
 m_{2}&= \frac{1}{24 t_4}(-12 a_1 m_1 - 12 a_2 m_1 - 12 b_1 m_1 - 12 b_2 m_1 - 4 t_2 - 
 3 a_1^2 t_4 \\
 &- 2 a_1 a_2 t_4 - 3 a_2^2 t_4 - 2 a_1 b_1 t_4 - 2 a_2 b_1 t_4 - 
 3 b_1^2 t_4 - 2 a_1 b_2 t_4 - 2 a_2 b_2 t_4 - 2 b_1 b_2 t_4 - 3 b_2^2 t_4).
 \end{split}
\end{align}
The remaining expressions are far more unwieldy.

However, assuming $\text{supp}(\rho) = [-b,-a]\cup [a,b]$, the formulae for the density function and moments become simple, as first found in \cite{khalkhali2020phase}. We have that
$$m_{2} =-\frac{t_{2}}{8t_{4}},$$
$$m_{4} = -\frac{t_{2}^2+ 8t_{4}}{64 t_{4}^{2}},$$
and
\begin{align*}
\begin{split}
\rho(x) = \frac{4t_{4}}{\pi}|x| \sqrt{(x^{2}-a^{2})(b^{2}-x^{2})}_{[-b,-a]\cup [b,a]},
\end{split}
\end{align*}
where the limit points of the support $[-b,-a]\cup [b,a]$ are
\begin{equation*}\label{a in terms of g}
a^{2}= \frac{-t_{2}+ 4 \sqrt{2t_{4}}}{8t_{4}},
\end{equation*}
and 
\begin{equation*}\label{b in terms of g}
b^{2}= -\frac{t_{2}+ 4 \sqrt{2t_{4}}}{8t_{4}}.
\end{equation*}
See Appendix \ref{AppA} for more details.

\subsection{The free energy}
Given a spectral density $\rho(x)$, we wish to compute the free energy. See Appendix \ref{appC} for details as well as the same computation carried out for an asymmetric one-cut case.

\begin{corollary}
\label{cor: sym free energy}
 The free energy for a one-cut symmetric solution of the type $(1,0)$ and $(0,1)$ quartic Dirac ensemble is of the form 
 \begin{align*}
 I[\rho_{\textit{1-cut}}^{\text{sym}}] &= -\ln (| \gamma | )+384 \gamma ^{12}+144 \gamma ^8+4 \gamma ^4+2 \left(8 \gamma ^6+\gamma ^2\right) t_{2}+\frac{1}{2},
 \end{align*}
 where $\gamma$ is a solution of \eqref{eq:1-cut tree equation}.

 The free energy for a symmetric two-cut solution of the type $(0,1)$ and $(1,0)$ quartic Dirac ensemble is 
 \begin{align*}
 I[\rho_{\text{2-cut}}^{\text{sym}}]=&\ \frac{1}{{4
 \left(\sqrt{-t_{2}-4 \sqrt{2}}+\sqrt{4 \sqrt{2}-t_{2}}\right)^2}}\bigg[t_{2}^3-t_{2}^2\sqrt{-t_{2}-4 \sqrt{2}} \sqrt{4 \sqrt{2}-t_{2}} \\
 &\left.+4
 \sqrt{-t_{2}-4 \sqrt{2}} \sqrt{4 \sqrt{2}-t_{2}}\right.\\
 &+8 t_{2} \ln
 \left(\frac{\sqrt{-t_{2}-4 \sqrt{2}}+\sqrt{4 \sqrt{2}-t_{2}}}{4 \sqrt{2}}\right)\\
 &+4
 \sqrt{-t_{2}-4 \sqrt{2}} \sqrt{4 \sqrt{2}-t_{2}} \left(\ln (32)-2 \ln
 \left(\sqrt{-t_{2}-4 \sqrt{2}}+\sqrt{4 \sqrt{2}-t_{2}}\right)\right)\bigg].
 \end{align*}
\end{corollary}

\subsection{Solutions of one matrix Dirac ensembles}

The stage is now set to obtain solutions for the two types of Dirac ensembles that are described by a single random Hermitian matrix with probability distribution \eqref{action}. Both are multi-trace models with combinations of moments $m_1,m_2$ and $m_3$ in the probability distribution. Therefore, the situation is going to be the same as in the case for the general pure potential single-trace quartic model, and we expect solutions with eigenvalue density supported on either one or two intervals. Since the action includes odd moments, we need to allow, at least in principle, for the ${H\to - H}$ symmetry of the model to be spontaneously broken. This was observed many times before for the models describing fuzzy field theories, see \cite{Tekel:2023tdx} for a review, and in a simpler model considered in detail recently in \cite{Bukor:2024kqy}. Hints of such behavior for these models appeared before \cite{d2022numerical} and were studied in \cite{d2026symmetry}.

As we have seen, the symmetric regime of these models is straightforward, since the solutions are completely determined by the saddle point equation. The model deforms the quadratic coupling. In the case of a one-cut solution, one obtains a polynomial equation of high order. If we are after the phase transition, things simplify further, and the deformation of the phase transition line of the pure potential model can be found analytically.

The asymmetric regime is more complicated. The one-cut distribution is unique, with defining equations that are not solvable analytically in a practical manner. We can solve them numerically and calculate the free energy given in the previous section.

The two-cut solution is a different can of worms. We can again solve (\ref{eq:2-cut C5}-\ref{eq:2-cut C7}) to isolate the saddle point conditions (\ref{eq:2-cut C1}-\ref{eq:2-cut C3}), but these do not uniquely determine the solution. There is a buffet of two-cut solutions which differ by the filling fraction, i.e. the fraction of the eigenvalues in one of the intervals
\begin{equation}
 \xi=\int\limits_{a_1}^{b_1} dx\,\rho(x)\in[0,1].
\end{equation}
The condition \eqref{eq:2-cut C4} is then equivalent to choosing the solution with minimal free energy, i.e. the solution where any small variation in the filling fraction leads to a decrease in probability.

When searching for the actual solutions, condition \eqref{eq:2-cut C4} is not practical. We have thus chosen a different approach\footnote{This was first used in \cite{Bukor:2024kqy}, but there the asymmetric two-cut solution was never preferred.}. For given values of the parameters of the model, we vary the filling fraction\footnote{In practice, we vary the ratio of the widths of the two intervals $(b_2-a_2)/(b_1-a_1)$, which is equivalent but technically easier to implement.} of the two-cut solution, numerically solve the defining equations, which have a unique solution for a given filling fraction, obtain the free energy and look for the value of the filling fraction that leads to minimal free energy. This process can be repeated several times depending on the desired precision.

For each parameter value, we need to compute the free energy of this solution. 

 \subsubsection{Solutions for the type $(0,1)$ model}
 
 In \cite{khalkhali2025large}, it is shown using Theorem 3.2 and Proposition 3.5 that there is a unique equilibrium measure that minimizes the free energy functional, and those solutions turn out to be symmetric. It has only two phases: a symmetric one-cut phase and a symmetric two-cut phase. When $t_{4} =1$ and $t_{2} < -4\sqrt{2}$, the symmetric solution on $[-b,-a]\cup [a,b]$ was found for both models in \cite{khalkhali2020phase}. What we find in this work is that even when one considers asymmetric solutions, the symmetric solutions are still the ones that minimize the free energy functional, agreeing with \cite{d2022numerical,d2026symmetry}. This transition is proven to be second order. It is also worth noting that symmetric solutions for the type $(0,1)$ and type $(1,0)$ models both have the same equilibrium measures. In Figure \ref{fig:type (0,1) free energy and moment}, we plot the free energy and second moment of the symmetric solutions.
\begin{figure}[H]
 
 \centering
 \begin{subfigure}[b]{0.45\textwidth}
 \centering
 \includegraphics[width=\textwidth]{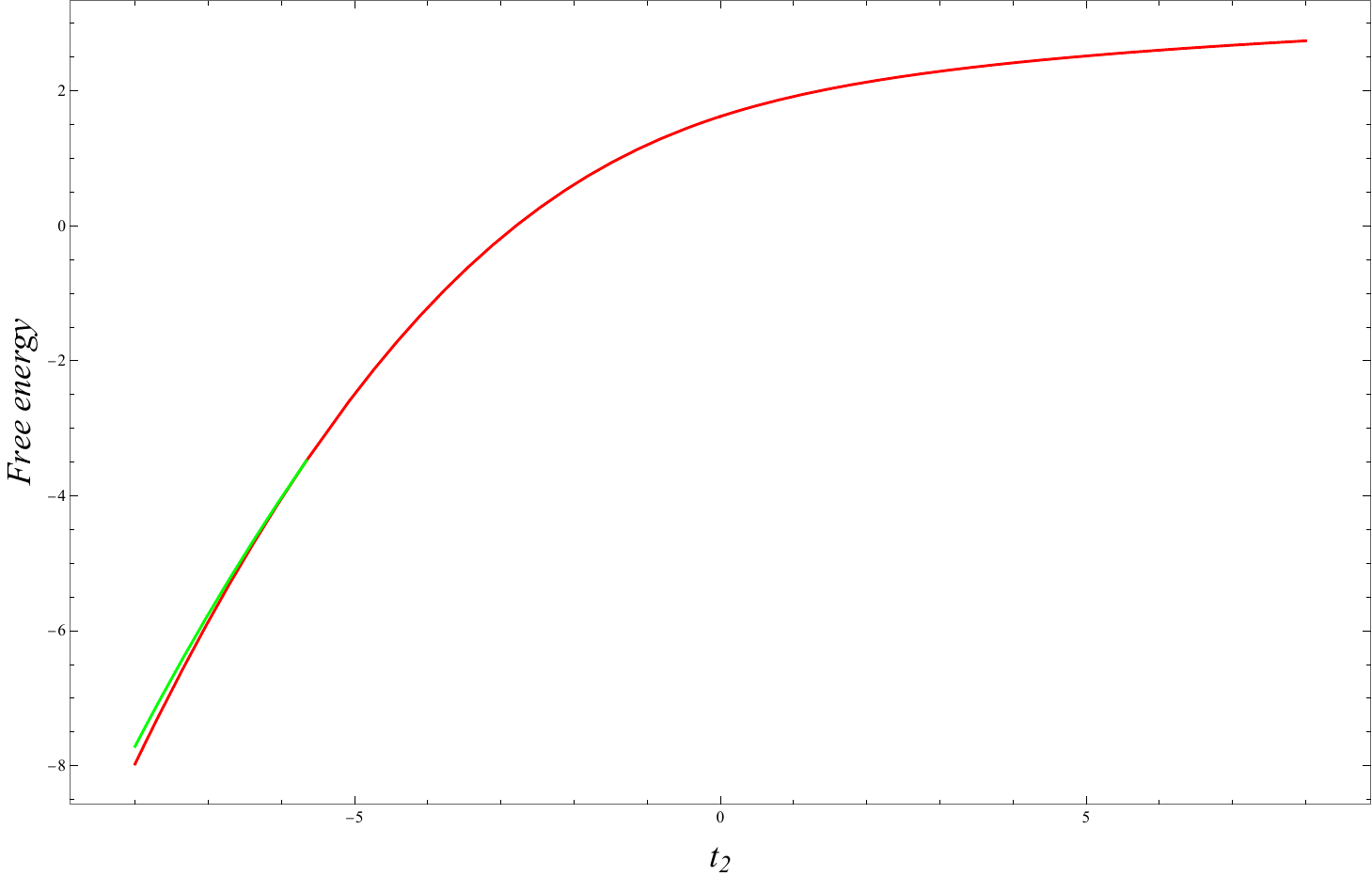}
 \label{fig:type (0,1) free energy}
 \end{subfigure}
 \begin{subfigure}[b]{0.45\textwidth}
 \centering
 \includegraphics[width=\textwidth]{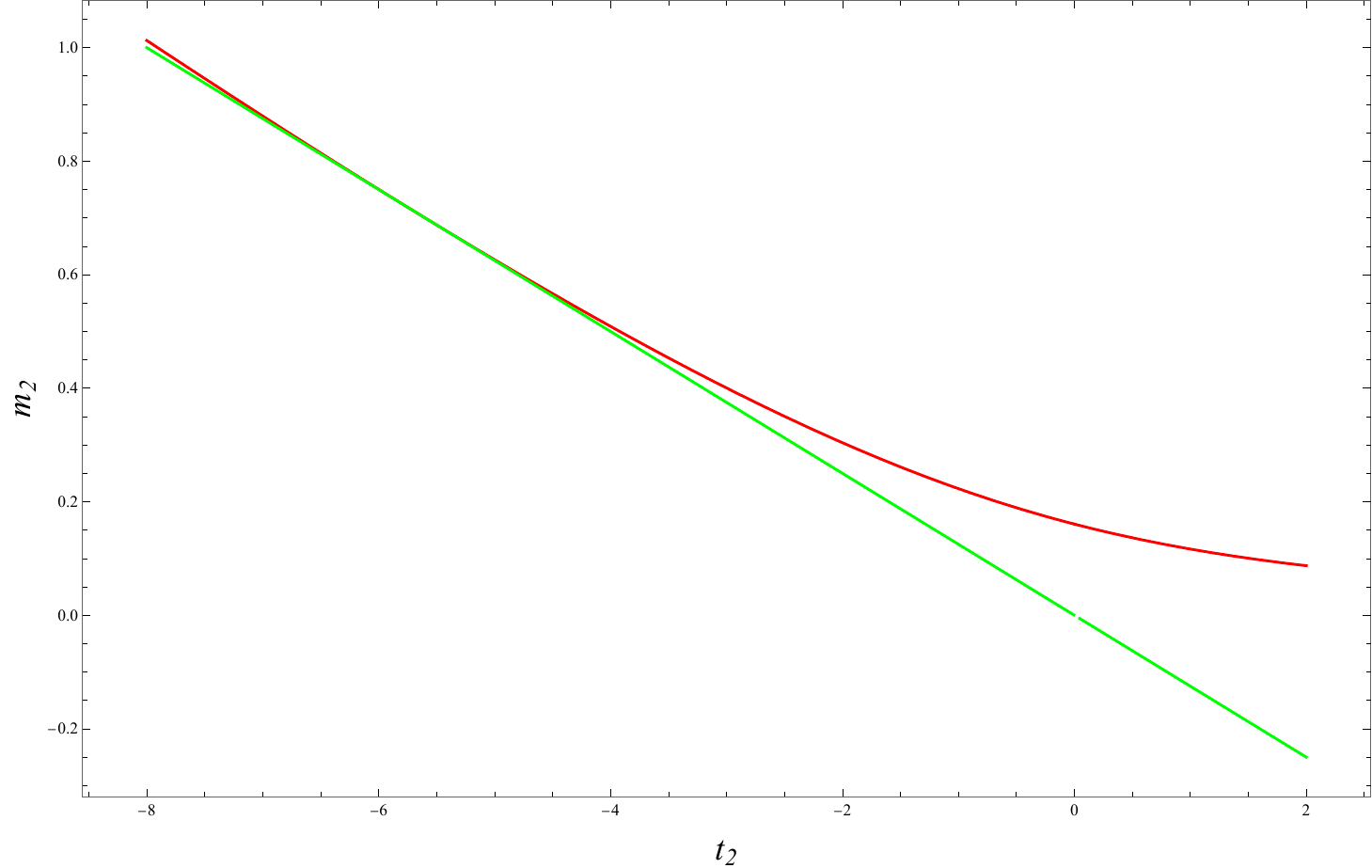}
 \label{fig:type (0,1) m2}
 \end{subfigure}
 \caption{The Figure on the left plots the free energy of the symmetric one-cut and two-cut solutions for $t_{4}=1$. The Figure on the right plots the symmetric one-cut and two-cut solutions' second moment for $t_{4}=1$. In both figures, the one-cut solution is in red and the two-cut solution is in green.}
 \label{fig:type (0,1) free energy and moment}
\end{figure}

 \subsubsection{Solutions for the type $(1,0)$ model}

 In this signature, the symmetric regime is the same, but the different signs in equation \eqref{action} have profound consequences. Following the procedure outlined in the introduction to this section, we have obtained the free energy diagram shown in Figure \ref{fig:type10free}.

 \begin{figure}[H]
 
 \centering
 \includegraphics[width=0.48\textwidth]{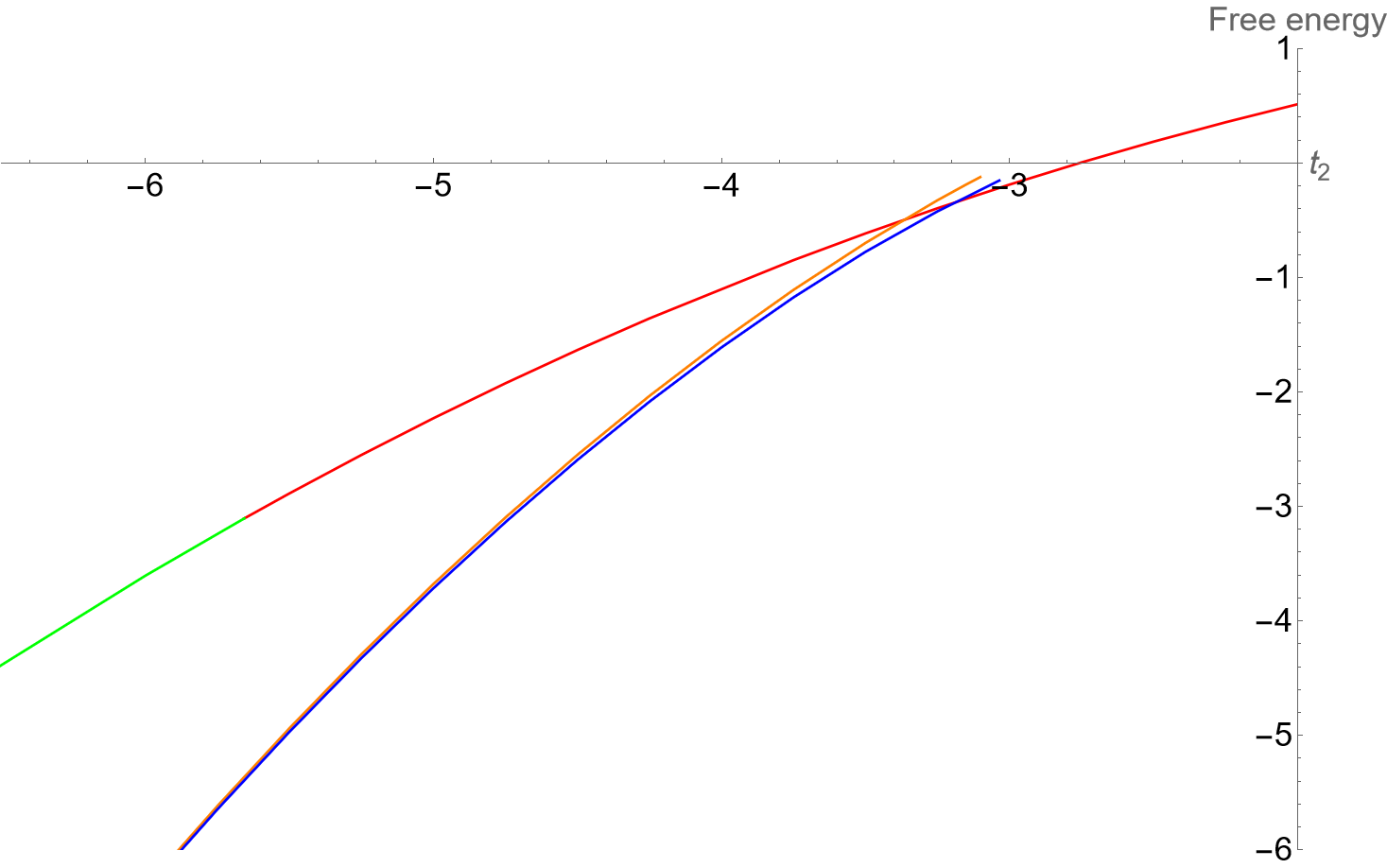}
 \includegraphics[width=0.48\textwidth]{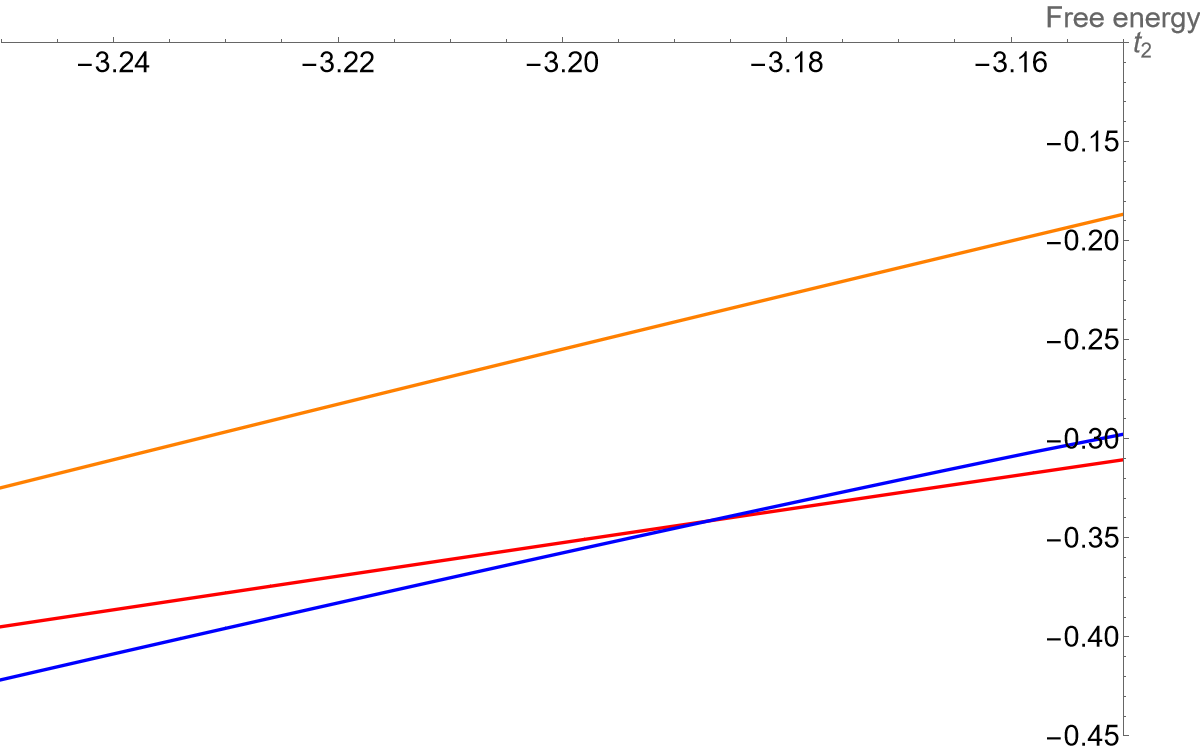}
 
 \caption{Comparison of free energies for various solutions of the $(1,0)$ model for ${t_4=1}$. The red line corresponds to the symmetric 1-cut solution, the green line to symmetric 2-cut solution, the blue line to the asymmetric 2-cut solution, and the orange line to the asymmetric 1-cut solution. The right panel zooms into the region close to the phase transition. In the left panel, we can see that the free energy of the asymmetric 1-cut and asymmetric 2-cut solutions come closer to each other as $t_2$ decreases, but an asymmetric 2-cut solution with some small, but nonzero filling fraction is always preferred. In the right panel, we can see that the phase transition happens before the value ${t_2\approx -3.036}$ where the asymmetric 2-cut solution ceases to exist.} 
 \label{fig:type10free}
\end{figure}

For positive and moderately negative $t_2$, the model stays in the symmetric regime. But at a critical $t_2$, for ${t_4=1}$ numerically determined to be
\begin{equation}\label{critt2}
 t_{2,c}=-3.18703
\end{equation}
in agreement with \cite{d2026symmetry}, the model starts to prefer the asymmetric two-cut solution with a very modest filling fraction. This happens before the one-cut to two-cut phase transition of the symmetric regime and thus there is no symmetric two-cut solution. We can see that the phase transition is more abrupt than in the symmetric regime, as the free energy does not change smoothly and its first derivative is discontinuous. This has also been the case for models describing fuzzy field theories \cite{Tekel:2023tdx}, but there the transition was between symmetric two-cut and asymmetric one-cut solutions.

To illustrate the situation, Figure \ref{fig:type10dens} shows the eigenvalue distribution of the preferred solution of the model for various values of the parameter $t_2$ on both sides of the critical point \eqref{critt2}.

 \begin{figure}[H] 
 
 \centering
 \includegraphics[width=\textwidth]{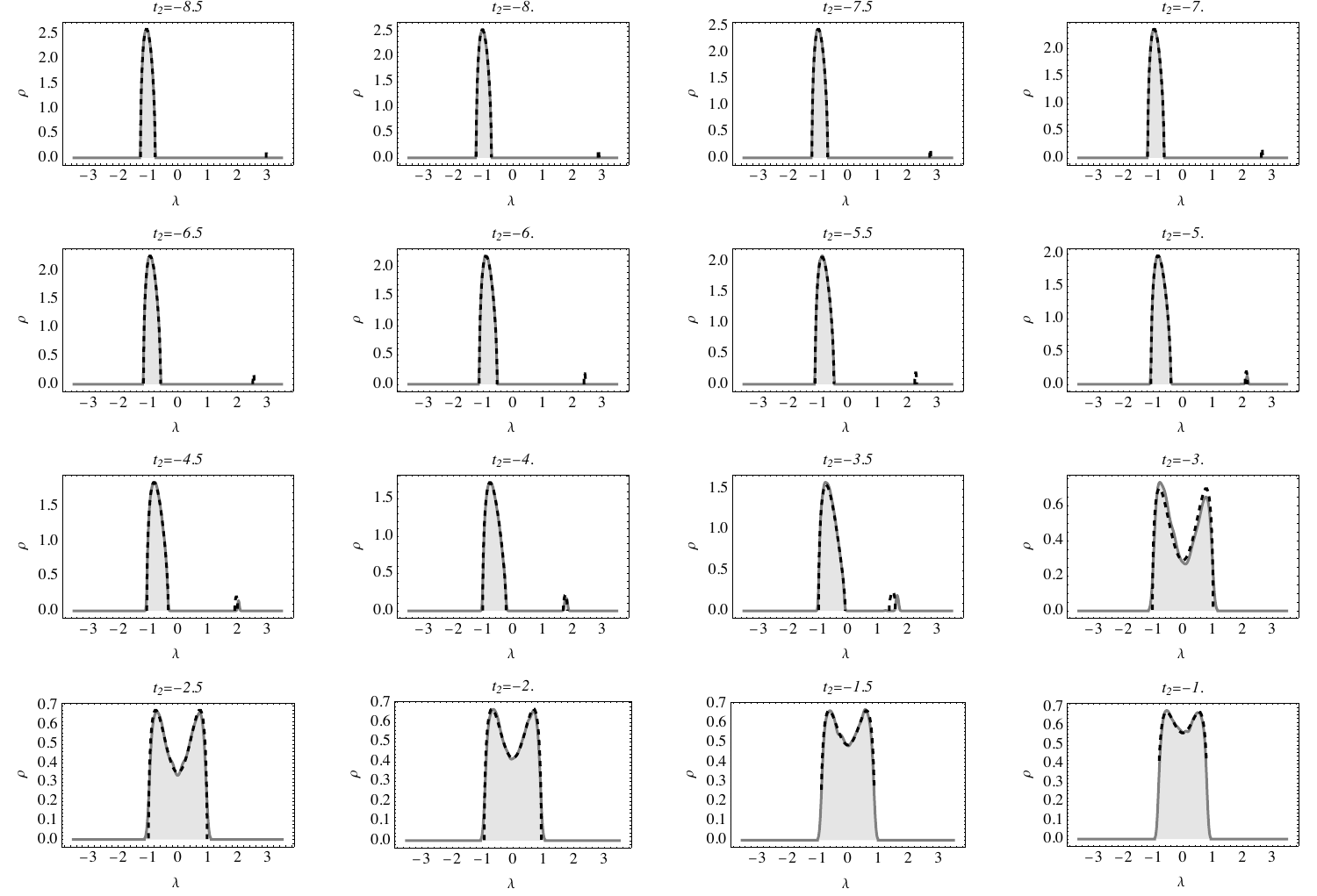}
 
 \caption{Comparison of the eigenvalue density function for various solutions of the $(1,0)$ model for ${t_4=1}$. We have used the cluster-flipping algorithm and aligned the eigenvalues at each step to ensure a negative sum, since the algorithm frequently flips between vacua of opposite signs. This enhanced the mild numerical asymmetry of the $t_2=-3$ result. The dashed line shows the theoretical prediction. Note that simulations failed to land in the two-cut solution.}
 \label{fig:type10dens}
\end{figure}

The consequences of this transition for the change in (spectral) geometry described by the corresponding Dirac operators, and especially the difference from the symmetric transition, are under investigation.

For different values of $t_4$, the situation is analogous. We have scanned the parameter space for different values of $t_4$ and repeated the numerical procedure, obtaining the phase diagram shown in the Figure \ref{fig:type10phase}.

 \begin{figure}[H]
 
 \centering
 \includegraphics[width=0.75\textwidth]{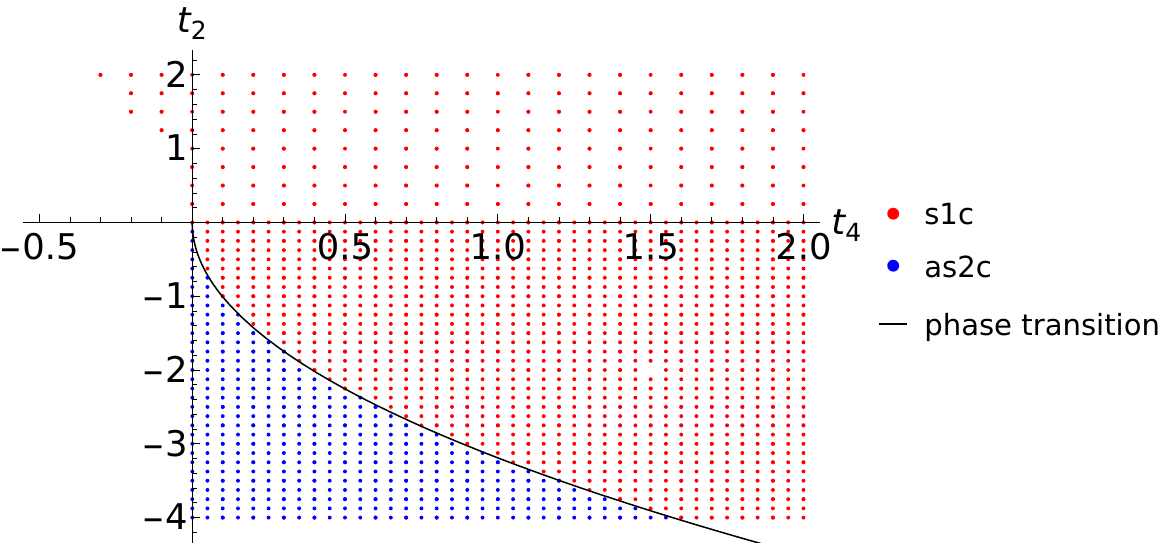}
 
 \caption{Phase diagram of the $(1,0)$ model in the ${(t_4,t_2)}$ parameter space. Blue dots represent parameters for which the preferred solution is of the one-cut type, the red dots represent parameters for which the preferred solution is of the asymmetric two-cut type. The black line is the fit of the transition line as a line separating the two regions given by \eqref{trafo}}
 \label{fig:type10phase}.
\end{figure}

There is a critical value of $t_2$ under which the stable solution is the asymmetric two-cut with a small but non-zero filling fraction, which decreases as $t_2$ becomes more negative, but never vanishes.

This was to be expected, since by rescaling the matrix $H$ in \eqref{action} one can set the value of $t_4$ to $1$ while changing the value of $t_2$ to $t_2/\sqrt{t_4}$. The best numerical fit of the line separating the two regions 
\begin{equation}\label{trafo}
t_2=- 3.19052 \sqrt{t_4}
\end{equation}
is consistent with this. We have also considered the case of positive $t_2$ and negative $t_4$, where the finite $N$ integral is only formal and the models make sense only in the $N\to\infty$ case. We confirmed that the asymmetric terms do not have any consequences here and we obtain only the symmetric one-cut solution for small enough $|t_4|$. At a certain point, the eigenvalues spill out of the local minimum created by positive $t_2$.

It is worth noting that there is another type of critical behavior not seen in the eigenvalue density functions in the same manner. In \cite{hessam2023double,khalkhali2024coloured}, the critical points associated with the double scaling limit of symmetric formal solutions of the Schwinger-Dyson equations for Dirac ensembles were studied. Such critical points are unrelated to those studied in this paper. It would be interesting to study whether asymmetric solutions of Dirac ensembles have new critical points in this sense.

\section{Bootstraps and reconstruction}\label{sec:bootstraps}
\subsection{The Schwinger-Dyson equations}
The Schwinger-Dyson equations (SDE) of these matrix models can be computed in a standard manner by considering the following identity 
\begin{equation*}
 \sum_{i,j=1}^{N}\int_{\mathcal{H}}\frac{\partial}{\partial H_{i j}}\left((H^{\ell+1})_{p q} e^{-S(D)} \right)dH = 0,
\end{equation*}
which follows from Stokes' Theorem. By rearranging the left-hand side and taking the large $N$ limit, one arrives at the following non-linear system of equations for $\ell>0$:
\begin{equation*}
 \sum_{k=0}^{\ell-1} m_k m_{\ell-k-1}=t_{2}\left(4 m_{\ell+1}+ 4 \epsilon m_1 m_{\ell}\right)+8t_{4}(m_{\ell+3}+ \epsilon m_3 m_{\ell}+ 3\epsilon m_1 m_{\ell+2}+3 m_2 m_{\ell+1}).
\end{equation*}
For more details we refer the reader to Section 3 of \cite{hessam2023double}. The formal symmetric solutions of the SDE were found from the SDE in the same work. An analogous approach can be applied for multi-matrix integrals where the entries of powers of $H$ are replaced by entries of words in the relevant matrix variables \cite{khalkhali2024coloured}.

\subsection{Positivity}
 In order to bootstrap these models, we first need positivity constraints.
Given a sequence of real numbers $(m_{n})$, the Hamburger moment problem asks if there is a positive Borel measure $\mu$ on the real line whose moments are precisely this sequence, i.e.
\begin{equation*}
 m_{n} = \int_{\mathbb{R}}x^{n}d\mu(x).
\end{equation*}
It is well-known that such a sequence exists if and only if the infinite Hankel matrix of moments
\[
\mathcal{H} =
\begin{pmatrix}
m_0 & m_1 & m_2 & m_3 & \cdots \\
m_1 & m_2 & m_3 & m_4 & \cdots \\
m_2 & m_3 & m_4 & m_5 & \cdots \\
m_3 & m_4 & m_5 & m_6 & \cdots \\
\vdots & \vdots & \vdots & \vdots & \ddots
\end{pmatrix}
\]
is positive semi-definite. That is, for all finitely supported sequences of complex numbers $(c_{i})$, we have that 
\begin{equation*}
 \sum_{i,j=0}^{\infty} m_{i+j}c_{i}\overline{c_{j}} \geq 0.
\end{equation*}

In the context of random matrices, if we assume that a given random Hermitian matrix $H$ has a well-defined distribution of eigenvalues, which is a very natural assumption, then its Hankel matrix of tracial moments must be positive semi-definite:
\[
\mathcal{H} =
\begin{pmatrix}
1 & 
\frac{1}{N}\ex[\tr H] & \frac{1}{N}\ex[\tr H^2]& \frac{1}{N}\ex[\tr H^3] & \cdots \\
\frac{1}{N}\ex[\tr H] & \frac{1}{N}\ex[\tr H^2] & \frac{1}{N}\ex[\tr H^3] & \frac{1}{N}\ex[\tr H^4]& \cdots \\
\frac{1}{N}\ex[\tr H^2] & \frac{1}{N}\ex[\tr H^3] & \frac{1}{N}\ex[\tr H^4] & \frac{1}{N}\ex[\tr H^5] & \cdots \\
\frac{1}{N}\ex[\tr H^3] & \frac{1}{N}\ex[\tr H^4] & \frac{1}{N}\ex[\tr H^5] & \frac{1}{N}\ex[\tr H^6] & \cdots \\
\vdots & \vdots & \vdots & \vdots & \ddots
\end{pmatrix},
\]
where $\mathbb{E}[\cdot]$ denotes the expectation value with respect to \eqref{action}.

By taking truncated submatrices of the Hankel matrix, one can derive non-linear constraints on the moments. Bootstrapping with positivity amounts to finding the region of possible solutions for a finite number of the SDE with a finite number of constraints from the Hankel matrix. As one increases the number of SDE and constraints, the region of possible solutions should eventually converge to any actual solutions. In practice, any solutions are usually apparent after less than a dozen SDE and constraints. Note that for Dirac ensembles, one obtains multiple sets of constraints since both $D$ and $H$ matrices have their own positivity constraints. Bootstrapping random matrices with positivity constraints was first done by Lin \cite{lin2020bootstraps}. Since this original work, many matrix integrals have been bootstrapped \cite{li2025analytic,khalkhali2025bootstrapping,hessam2022bootstrapping}, and modified or additional methodologies have arisen \cite{kazakov2022analytic,Kovacik:2025qgj,toriumi2026,maeta2026matrix}. More generally, the idea of bootstrapping with positivity has gained use in other areas of theoretical physics such as matrix quantum mechanics \cite{han2020bootstrapping}. The symmetric bootstrapped solutions to our models of interest were found in \cite{hessam2022bootstrapping}. We will expand our search here to asymmetric solutions. Note that we restrict our plots to when $t_{4}=1$, since $t_{4}$ is a redundant coupling constant; see the Appendix D of \cite{khalkhali2024coloured}.

\subsection{Method of eigenvalue reconstruction}
Given the moments obtained from the bootstrap, a natural question is whether one
can reconstruct the corresponding eigenvalue distribution from a finite subset
of moments. The goal is to reconstruct the eigenvalue distribution with sufficient
accuracy to evaluate the free energy and identify the preferred solution.

We follow the method of \cite{tekel2012constructing}, based on an expansion in orthogonal
polynomials with respect to a chosen initial distribution. Let $P(x)$ be an
initial distribution with moments $m_n$, and let $P'(x)$ denote the distribution
with moments $m_n'$ obtained from the bootstrap; details are in \cite{Kovacik:2025qgj}.

From the moments $m_n$, one constructs polynomials $K_n(x)$ orthogonal with
respect to $P(x)$,
\begin{equation}
\int K_n(x) K_m(x) P(x)\, dx = N_n \delta_{nm}.
\end{equation}
The distribution $P'(x)$ is then constructed as a deformation of $P(x)$ and
expanded in this set as
\begin{equation}
P'(x) = P(x) \sum_{n=0}^{M} \frac{1}{N_n}
\langle K_n(x) \rangle' K_n(x),
\end{equation}
where the coefficients are determined by the moments $m_n$ and $m_n'$. Since
$K_n(x)$ are polynomials, each coefficient depends only on a finite number of
moments. The truncation order $M$ controls the accuracy of the approximation.
A more detailed discussion of the reconstruction procedure in the context of
matrix models can be found in \cite{Kovacik:2025qgj}.

The convergence of this method depends sensitively on the choice of the initial
distribution, in particular on its support. To determine the support directly
from the moments, we use a shifted Hankel construction. Defining the Hankel
matrices
\begin{equation}
(H_0)_{ij}=m_{i+j}, \qquad (H_1)_{ij}=m_{i+j+1},
\end{equation}
the matrix $H_1$ corresponds to multiplication by $x$ in the monomial basis.
Introducing the truncated polynomial space
\begin{equation}
V_n=\mathrm{span}\{1,x,\dots,x^{n-1}\},
\end{equation}
one obtains the finite-dimensional representation
\begin{equation}
M_n = H_0^{-1}H_1
\end{equation}
of the multiplication operator projected onto $V_n$.

In the infinite-dimensional limit, the spectrum of the multiplication operator
coincides with the support of the measure. The eigenvalues of $M_n$ therefore
provide a finite-dimensional approximation to the support, with the extremal
eigenvalues converging to the endpoints as $n$ increases.

This determines the interval for the initial distribution and improves the
stability of the reconstruction. The resulting eigenvalue distribution is then
used to evaluate the free energy and compare different solutions.

\subsection{Type $(0,1)$}
As discussed in the previous sections, the solution space of the type $(0,1)$ model is far less complicated than the $(1,0)$. We find this in our bootstrap estimates. 
\begin{figure}[H]
 \centering
 \begin{subfigure}[b]{0.40\textwidth}
 \centering
 \includegraphics[width=\textwidth]{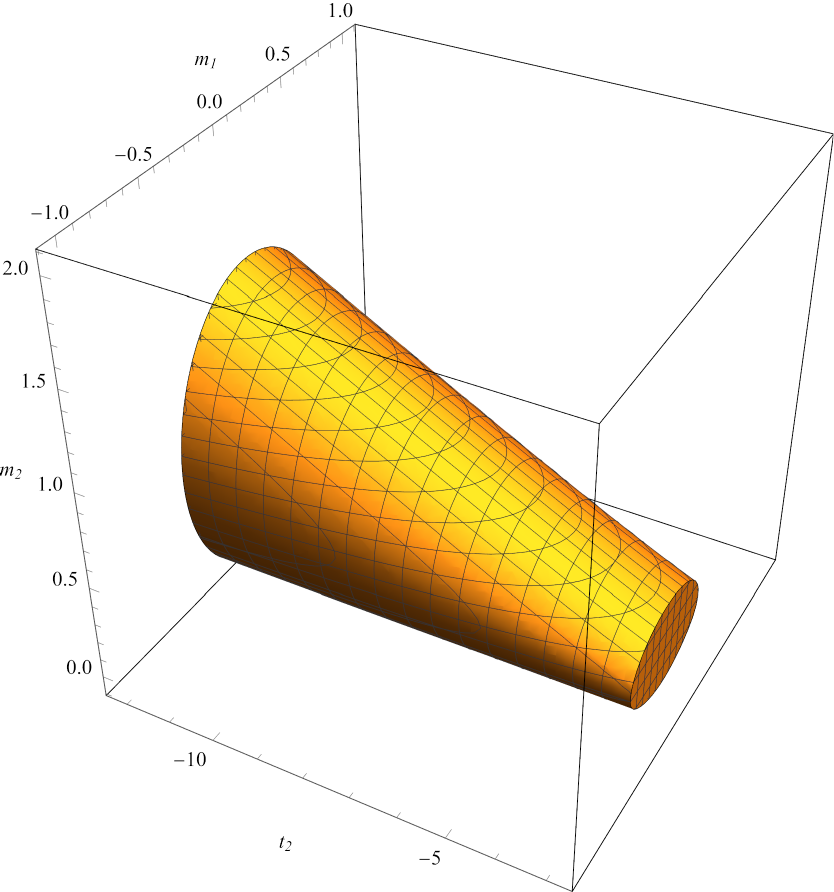}
 \label{fig:y equals x}
 \end{subfigure}
 \begin{subfigure}[b]{0.45\textwidth}
 \centering
 \includegraphics[width=\textwidth]{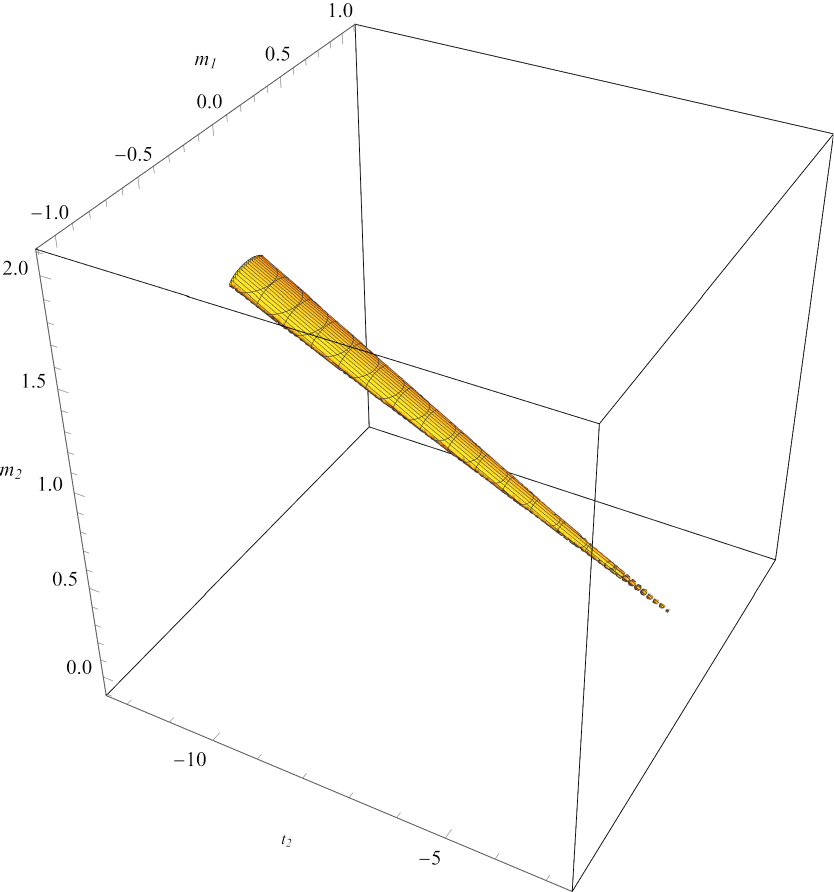}
 \label{fig:five over x 2}
 \end{subfigure}
 \caption{Bootstrapped region of solution space for the type $(0,1)$ quartic Dirac ensemble. The left plot was found with four SDE and a submatrix of the Hankel matrix of size four. The right plot was generated with ten SDEs and a Hankel submatrix of size ten.} 
 \label{fig:type (0,1) boostraps 3d}
\end{figure}

In particular, we see excellent agreement between the known symmetric solutions for a relatively small number of equations and matrix size. Even though we have allowed for non-zero odd moments, the same phase transition found in the purely symmetric case \cite{khalkhali2020phase,hessam2022bootstrapping} is found here. That is, there is a transition between the two-cut and one-cut symmetric solutions at $t_{2} =-4\sqrt{2}$ (when $t_{4}=1$). 

\begin{figure}[H]

 \centering
 \includegraphics[width=.75\textwidth]{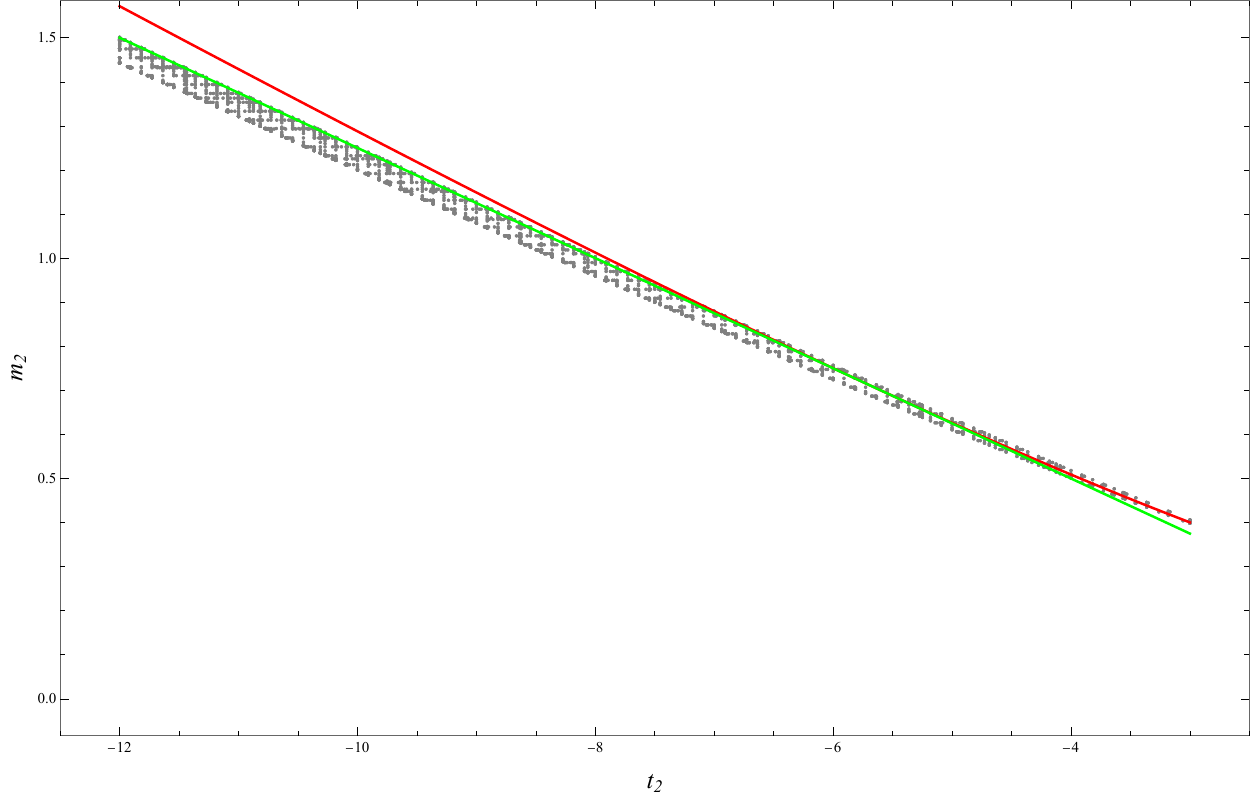}
 \caption{Bootstrapped points of the positive solution space for $m_{2}$ of the type $(0,1)$ quartic Dirac ensemble plotted against the known analytic symmetric solutions for $t_{4}=1$. Bootstrapping data was generated with ten SDE and a submatrix of the Hankel matrix of size ten. The real symmetric one-cut solution is plotted in red and the symmetric two-cut solution is plotted in green. There is a phase transition at $t_{2} = -4\sqrt{2}$ between these two solutions.} 
 \label{fig:type (0,1) boostraps} 
\end{figure}

\subsection{Type $(1,0)$}
Unlike in the type $(0,1)$ model, we see the appearance of two asymmetric two-cut solutions. These asymmetric solutions bypass the symmetric two-cut solution entirely. We find that there is an abrupt transition between the asymmetric two-cut solution and the symmetric one-cut solution. Both asymmetric solutions differ by a sign of the first moment. In Figure \ref{fig:type (1,0) boostraps 3d}, these solutions present as a fork in the solution space. The precise location of this phase transition is explored further in the next section.

\begin{figure}[H]
 \centering
 \begin{subfigure}[b]{0.45\textwidth}
 \centering
 \includegraphics[width=\textwidth]{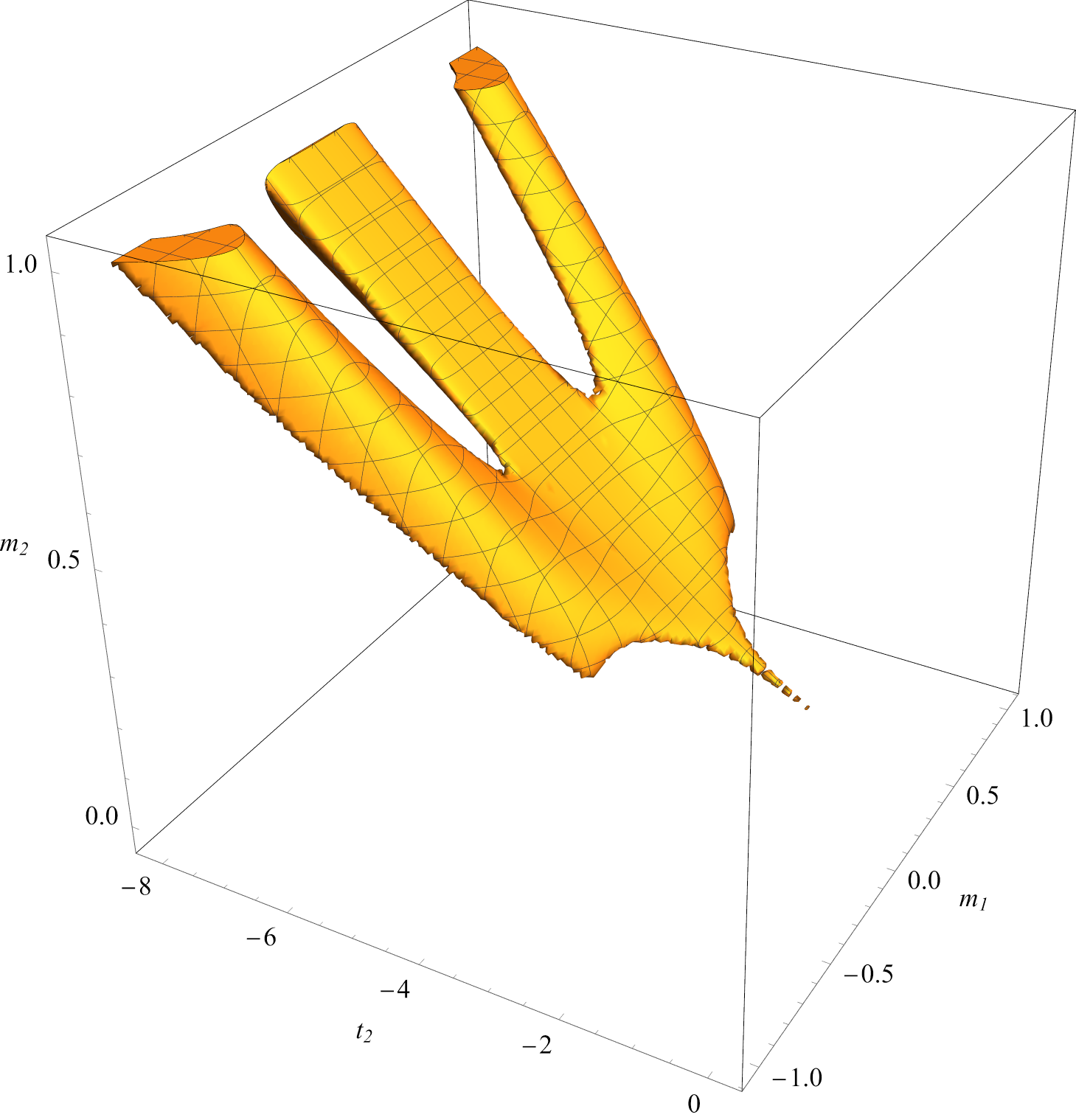}
 \label{fig:y equals x 2}
 \end{subfigure}
 \begin{subfigure}[b]{0.45\textwidth}
 \centering
 \includegraphics[width=\textwidth]{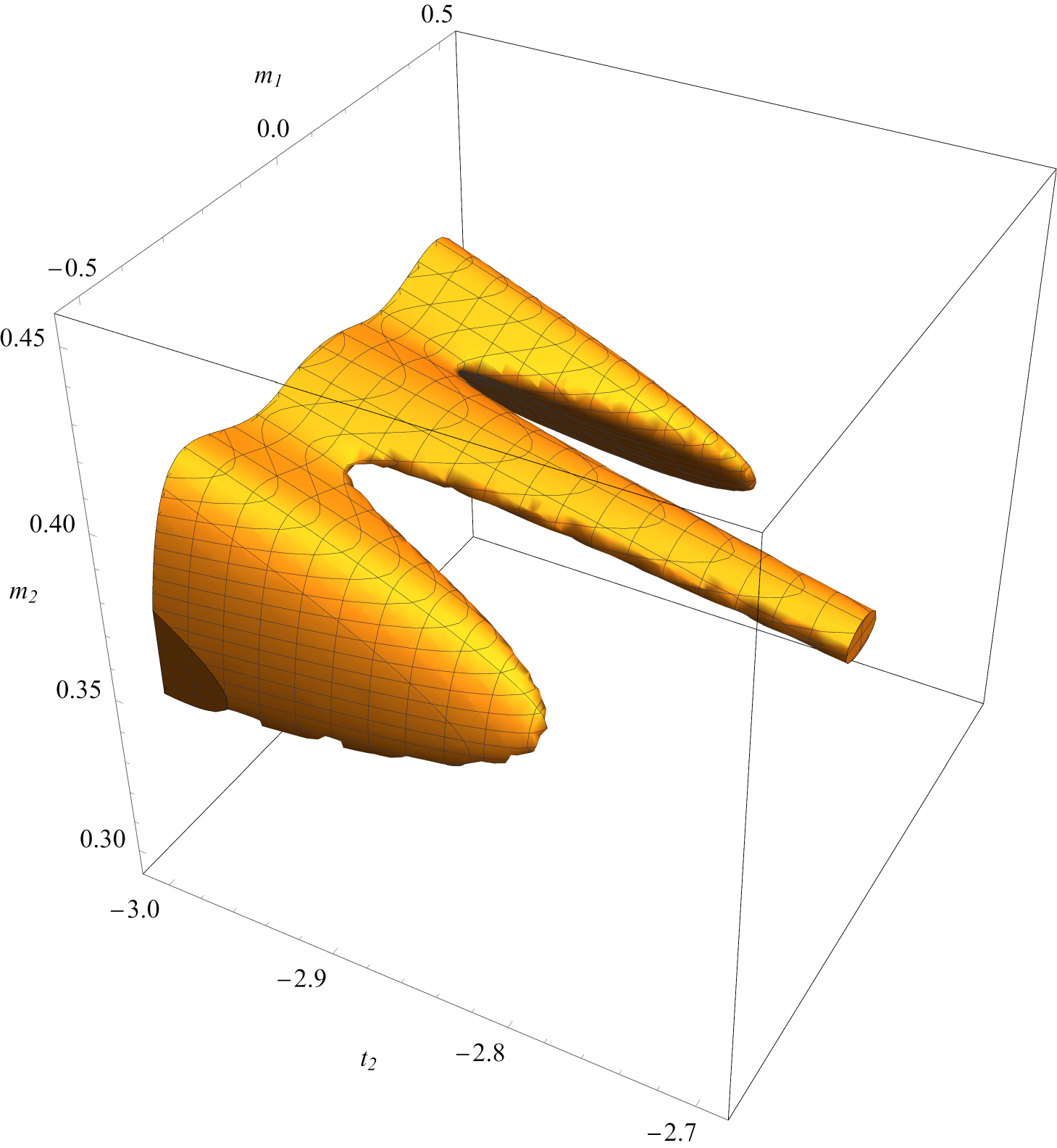}
 \label{fig:five over x}
 \end{subfigure}
 \caption{The bootstrapped region of the solution space for the type $(1,0)$ quartic Dirac ensemble, at different scales. These were found using 10 SDE and a submatrix of the Hankel matrix of size 10. We see a fork in the solution space indicating at least three possible solutions satisfy the positivity constraints and SDE.}
 \label{fig:type (1,0) boostraps 3d}
\end{figure}

\begin{figure}[H]
 \centering
 \begin{subfigure}[b]{0.45\textwidth}
 \centering
 \includegraphics[width=\textwidth]{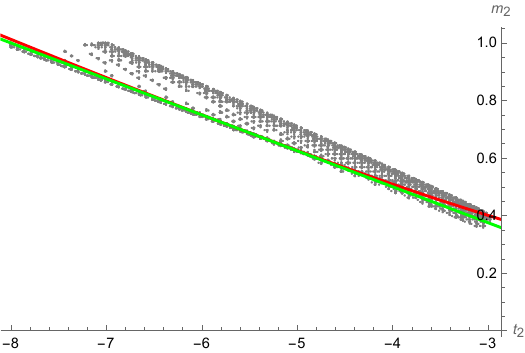}
 \end{subfigure}
 \begin{subfigure}[b]{0.45\textwidth}
 \centering
 \includegraphics[width=\textwidth]{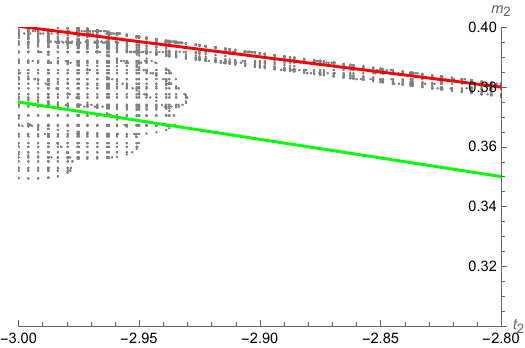}
 \end{subfigure}
 \caption{Bootstrapped points of the positive solution space for $m_{2}$ of the type $(1,0)$ quartic Dirac ensemble plotted against the known analytic symmetric solutions for $t_{4}=1$. The real symmetric one-cut solution is plotted in red and the symmetric two-cut solution is plotted in green. The left plot was found with twelve SDE and a submatrix of the Hankel matrix of size twelve. Here, we can see a clear separation between the asymmetric and symmetric solutions. The right plot was found with fourteen SDE and a submatrix of the Hankel matrix of size fourteen. We can see that when we refine the search, the two-cut symmetric solution can be ruled out. } 
 \label{fig:type (1,0) boostraps} 
\end{figure}

\begin{figure}[H]
 \centering
 \begin{subfigure}[b]{\textwidth}
 \centering
 \includegraphics[width=\textwidth]{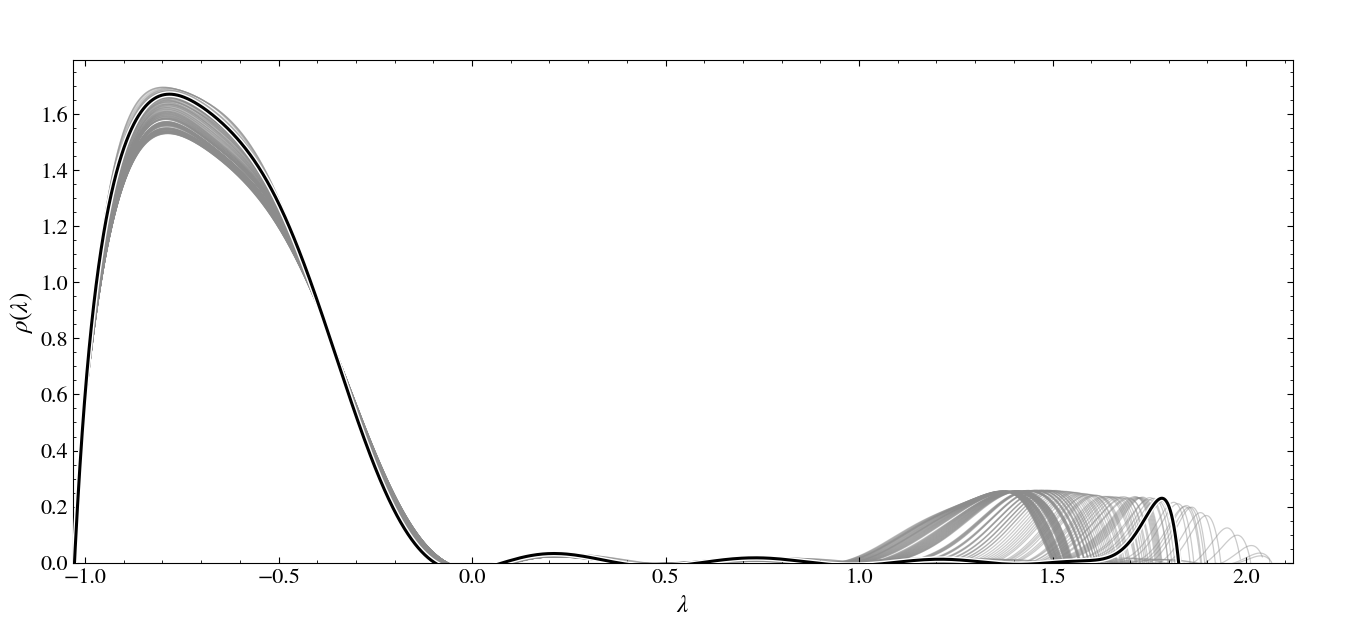}
 \end{subfigure}
 \caption{Two-cut asymmetric eigenvalue distributions reconstructed from the bootstrap
moments for $t_2=-4$. By reconstructing the distributions across different values
of $t_2$, we determine the energetically preferred solutions and compare the
symmetric and asymmetric phases. The preferred solution for this value of $t_2$
is highlighted in bold.}
 \label{fig: boot reconstruct} 
\end{figure}

\section{Numerical simulations}\label{sec:MC}

We have performed Hamiltonian HMC \cite{brooks2011handbook, Duane:1987de} of the action \eqref{action} for both the $(1,0)$ and $(0,1)$ models. The first $10\%$ of the simulation was thermalisation, then we measured eigenvalues every 10 steps, reducing the autocorrelation of the data. We tuned the step parameter to achieve an acceptance rate somewhat below $100\%$. This might seem resource wasteful, as the rate close to $70\%$, which is recommended, could be achieved by some fine-tuning. However, we wanted to use the same step across the whole span of $t_2$, so we took a conservative value of the step-length. 

\subsection{Simulations of the $(1,0)$ model}

From a certain point of view, these numerical simulations are standard in the field of matrix models and we do not explain in detail the technique. However, there is one crucial point which made this model stand out and forced us to develop a new algorithm to overcome the novel difficulties. The details can be found in \cite{Kovacik:2026ybz}. The issue with the model is that it has a rich structure of false vacua, for example for some values of $t_2$ we have one vacuum in which all of the eigenvalues are close to each other. Then, there are other vacua in which some small (and different) portions of eigenvalues are split from the main bulk and are positioned far away from it. Using the standard HMC, one way is to initiate the simulation very close to the true vacuum, as was done in \cite{d2026symmetry}. Another option is to start numerical simulation from different initial configurations and then pick the correct one by computing the corresponding free energy. This is doable to some extent and does not rely on prior knowledge of the model; however, it is both computationally expensive and, for models with multiple vacua, the true one can be missed. We have devised an algorithm that occasionally takes a portion of eigenvalues and moves them, en masse, to a different position – subject to the Metropolis check. Therefore, after sufficiently long simulation time, although significantly faster if left for spontaneous tunnelling, the system finds the true vacuum state. 

The system is best described by the eigenvalue distribution. This can be used to compute various observables, for example the first two moments of the distribution which are standard observables in matrix models:
\begin{equation} \label{eq:HMC_obs}
 m_1 =\ex \left[ \frac{1}{N}\sum \limits_{i=1}^N \lambda_i \right], 
 m_2 = \ex \left[ \frac{1}{N}\sum \limits_{i=1}^N \lambda_i^2 \right ].
\end{equation}

As the models are symmetric with respect to $\lambda \rightarrow -\lambda$, we plot $|m_1|$ instead of $m_1$ as this symmetry is spontaneously broken. The results for the $(1,0)$ model are shown in Figure \ref{fig:dirac10_observables}. We have performed the simulations for $N=10$ and $N=100$ to grasp the difference that the finite-N effects create. One can observe that the behaviour changes quantitatively but not qualitatively. 

\begin{figure}[htbp]
 \centering
 \begin{subfigure}[b]{0.48\textwidth}
 \includegraphics[width=\textwidth]{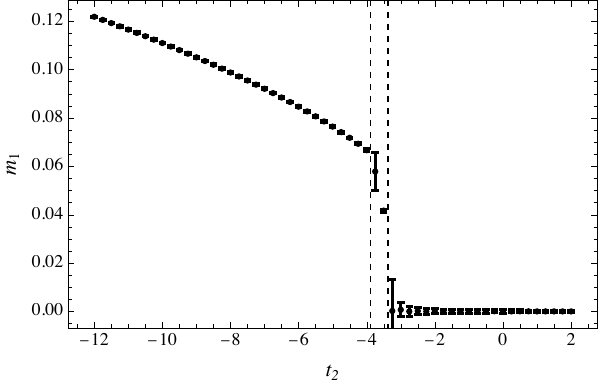}
 \caption{$|m_1|$, $N=10$.}
 \end{subfigure}
 \hfill
 \begin{subfigure}[b]{0.48\textwidth}
 \includegraphics[width=\textwidth]{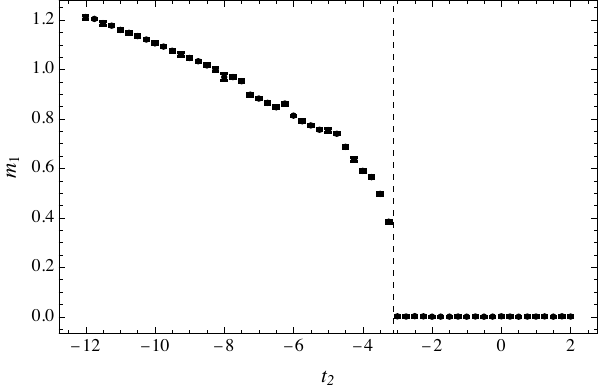}
 \caption{$|m_1|$, $N=100$.}
 \end{subfigure}

 \vspace{0.5em}

 \begin{subfigure}[b]{0.48\textwidth}
 \includegraphics[width=\textwidth]{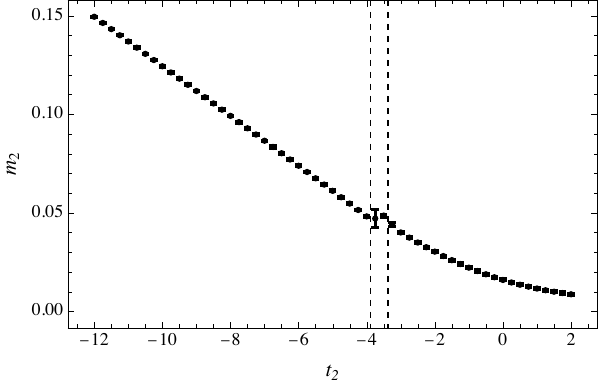}
 \caption{$m_2$, $N=10$.}
 \end{subfigure}
 \hfill
 \begin{subfigure}[b]{0.48\textwidth}
 \includegraphics[width=\textwidth]{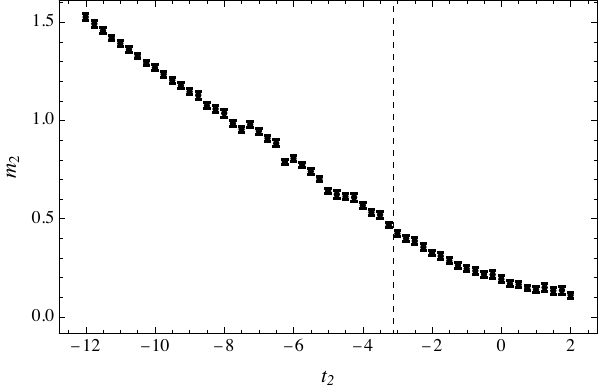}
 \caption{$m_2$, $N=100$.}
 \end{subfigure}

 \caption{Mean values of the observables \eqref{eq:HMC_obs} obtained using HMC simulations. The vertical dashed lines are at $t_2=-3.875$ and $t_2=-3.375$ for $N=10$ and at $t_2=-3.125$ for $N=100$. These points show where the behaviour of the eigenvalue distribution changes. The small jumps on the $N=100$ plots are explained in the caption of the Figure \ref{fig:dirac10_trajectories}.}
 \label{fig:dirac10_observables}
\end{figure}

Even using the aforementioned procedure, the simulation did not always find the preferred vacuum as the number of HMC steps was limited. While for $N=10$, it was easy to thermalise to the correct vacuum, in some of the $N=100$ runs, we noticed jumps during the main runs. The cusps are caused by different simulation settings across vacua; see Figure \ref{fig:dirac10_trajectories}.

\begin{figure}[htbp]
 {\centering
 \begin{subfigure}[b]{0.32\textwidth}
 \includegraphics[width=\textwidth]{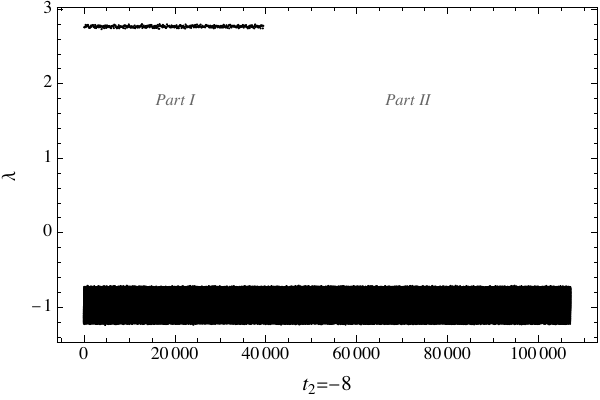}
 \caption{HMC trajectory of eigenvalues for $t_2=-8$.}
 \end{subfigure}
 \hfill
 \begin{subfigure}[b]{0.32\textwidth}
 \includegraphics[width=\textwidth]{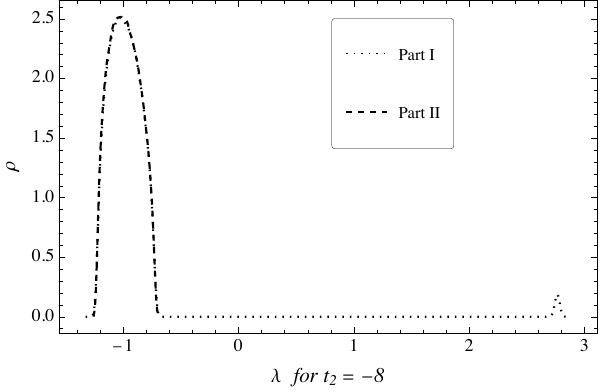}
 \caption{Eigenvalue distributions as measured in Part I and Part II.}
 \end{subfigure}
 \hfill
 \begin{subfigure}[b]{0.32\textwidth}
 \includegraphics[width=\textwidth]{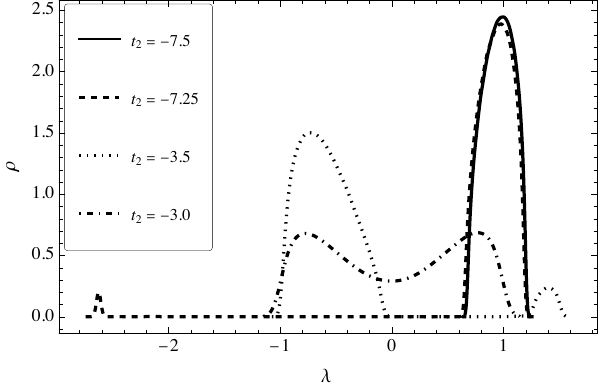}
 \caption{Eigenvalue distribution as measured at different values of $t_2$.}
 \end{subfigure}

 \caption{Panels (a) and (b) show the results (eigenvalue trajectories and distributions) for a particular simulation run for $t_2=-8$. We see that the algorithm found a different vacuum around the step $40, 000$. The corresponding distributions obtained by taking only the part before or after this transition lead to different distributions whose existence is confirmed by other methods. We have also measured the free energy before and after this transition; it decreased from $-13.229$ to $-13.246$. The right plot shows the eigenvalue distribution at some values of $t_2$. Notice that two close values of $t_2$ result in different distributions, not because there is a phase transition between them; the simulations just remained in different vacua and did not transition. }
 \label{fig:dirac10_trajectories}
 }
\end{figure}

\subsection{Simulations of the $(0,1)$ model}

Analyses of the $(0,1)$ model are simpler, mostly due to the fact that the vacuum is less structured. There are only symmetric solutions detected in the numerical analysis, and the model continuously transfers from a one-cut distribution for $t_2>-4\sqrt{2}$ to a two-cut distribution for $t_2<-4\sqrt{2}$. This model has one special feature: the action does not depend on the first moment, $\tr H$. In other words, the model has a symmetry $H \rightarrow H+\alpha \textbf{1}$. One way to address this issue would be to fix the code to consider only traceless matrices. The other, which we opt for, would be to use the same code but after measuring the eigenvalues, centre them at each Monte-Carlo step; that is, at each step, the stored eigenvalues are shifted to have $\sum \limits_{i=1}^N \lambda_i = 0$. Because of this, we show only the plot of $m_2$ and some sample eigenvalue distributions. 

\begin{figure}[htbp]
 {\centering
 \begin{subfigure}[b]{0.32\textwidth}
 \includegraphics[width=\textwidth]{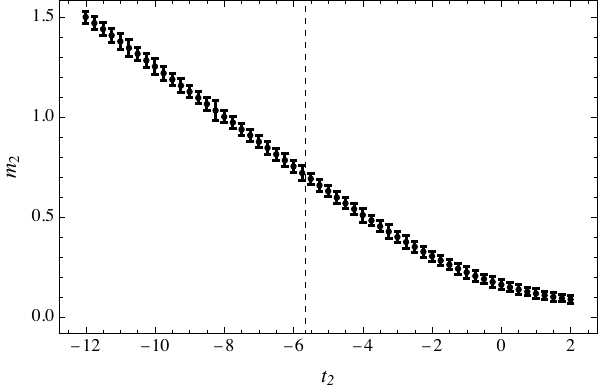}
 \caption{Mean values of $m_2$ for the $(0,1)$ model with $N=10$ with centred eigenvalues to have $m_1=0$. The vertical line is $t_2=-4 \sqrt{2}$.}
 \end{subfigure}
 \hfill
 \begin{subfigure}[b]{0.32\textwidth}
 \includegraphics[width=\textwidth]{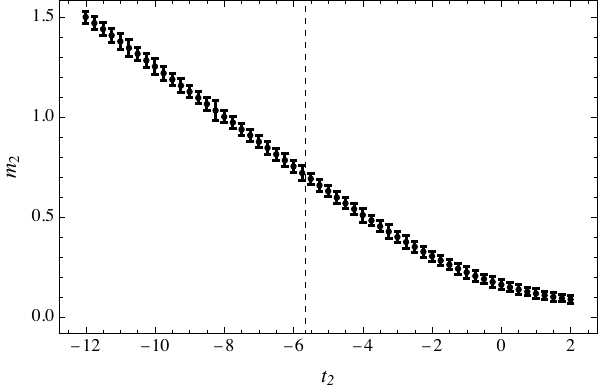}
 \caption{Mean values of $m_2$ for the $(0,1)$ model with $N=100$ and centred eigenvalues to have $m_1=0$. The vertical line is $t_2=-4 \sqrt{2}$.}
 \end{subfigure}
 \hfill
 \begin{subfigure}[b]{0.32\textwidth}
 \includegraphics[width=\textwidth]{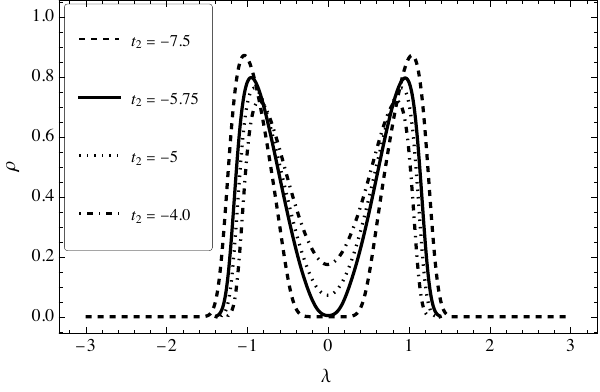}
 \caption{Eigenvalue distribution as measured at different values of $t_2$, the simulation closest to the critical value is the one with $t_2=-5.75$.}
 \end{subfigure}

 \caption{Results of the HMC simulation of the type $(0,1)$ quartic model showing a smooth transition around the critical value of $t_2=-4 \sqrt{2}$ where the distribution changes from the symmetric one-cut to the symmetric two-cut phase.}
 \label{fig:dirac01_results}
 }
\end{figure}

\section{Comparison with the analytical and bootstrap results}
\label{sec:BS}

We have presented three different ways to obtain the eigenvalue distribution for both the $(1,0)$ and $(0,1)$ model. The first is the standard analytical approach, the second is the bootstrap estimate and the third is the HMC numerical simulation. The analytical method is the most accurate in the current setting. Given bootstrapped results, we then use the method of \cite{Kovacik:2025qgj} to reconstruct the distribution. This step is crucial not only for making the plots but also for computation of the free energy which distinguishes between the plethora of vacua of the $(1,0)$ model. We start with a rather crude approximation so the bootstrap eigenvalue estimates contain fluctuations in the areas out of the support. The numerical simulations for $t_2=-3.15$ produced a clear result. For the $t_2=-4$ simulation, we have turned off the cluster algorithm and initiated the simulation in which $1-3$ of $100$ eigenvalues separated from the bulk and move to $-1.7,-1.75,-1.8,-1.85$ as an initial position; that is $12$ different initial conditions. The system remained in the selected solution, then the free energy was calculated, choosing a minimizing configuration. In the presented case, two eigenvalues moved to $-1.85$ (other shifts with two eigenvalues were rather similar). Figure \ref{fig:comparisons10} shows the comparison. Apart from the fluctuations of the bootstrap estimate, we see excellent agreement. Such fluctuations could be removed by choosing a different form of the reconstructing function. The barely visible shift of the smaller part of the distribution obtained using HMC for $t_2=-4$ might be caused by the fact that the true vacuum state would prefer to have more than two but less than three eigenvalues split. To achieve this, one would require a matrix size larger than $N=100$. Still, the agreement is satisfactory and we can claim to have the model under control in three different ways.

\begin{figure}[htbp]
 {\centering
 \begin{subfigure}[b]{0.48\textwidth}
 \includegraphics[width=\textwidth]{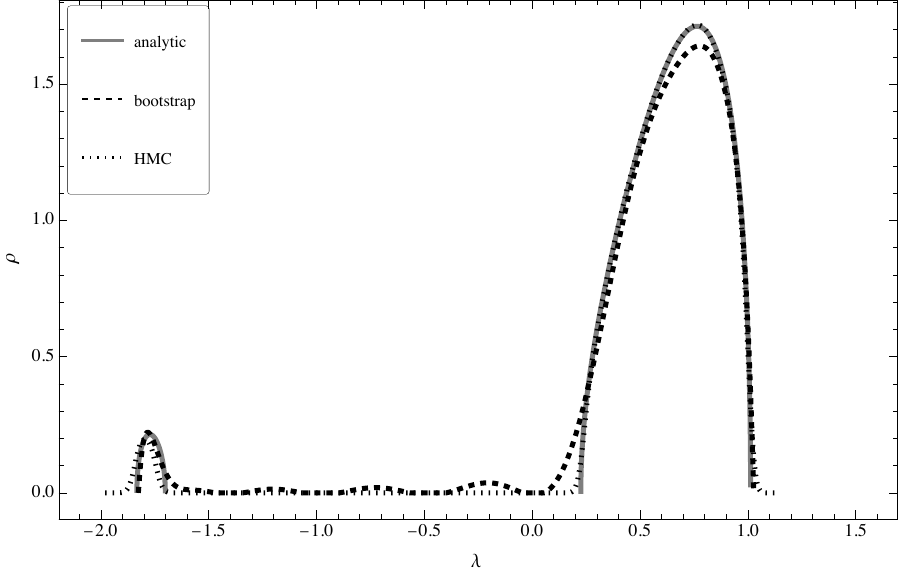}
 \caption{Comparison for $t_2=-4$.}
 \end{subfigure}
 \hfill
 \begin{subfigure}[b]{0.48\textwidth}
 \includegraphics[width=\textwidth]{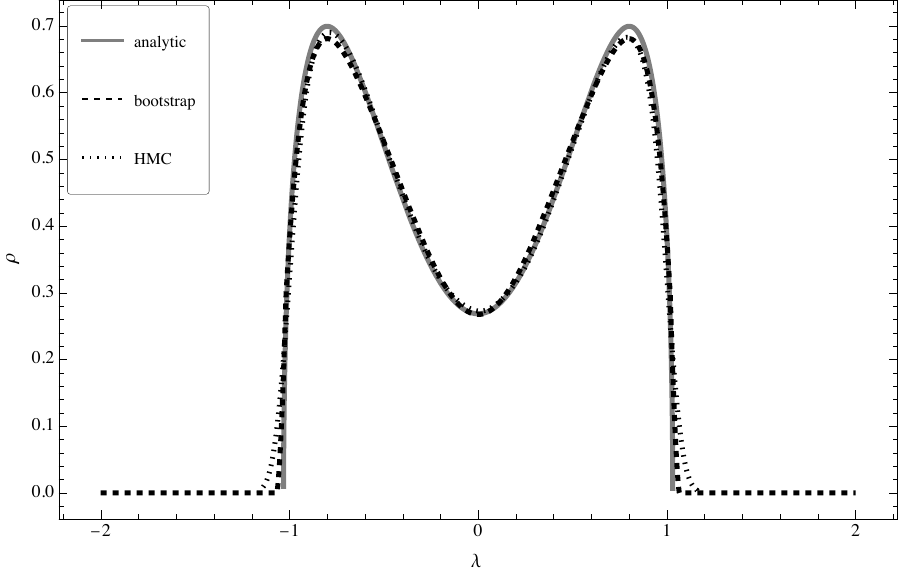}
 \caption{Comparison for $t_2=-3.15$.}
 \end{subfigure}

 \caption{Comparison of the eigenvalue estimates obtained using the analytical, bootstrap and HMC method with $N=100$ for the $(1,0)$ model. }
 \label{fig:comparisons10}
 }
\end{figure}

\section{Conclusions and future work}\label{sec:conclusion}
In this work, we studied the type $(1,0)$ and $(0,1)$ quartic Dirac ensembles using three approaches. We first evaluated these models by analyzing the associated Riemann-Hilbert problems, i.e. the saddle-point equations, with a focus on asymmetric solutions. We found explicit formulae for the distribution of eigenvalues and free energy for solutions of the saddle-point equation. The second approach applies bootstrapping with positivity to derive explicit bounds on the moments of these models. From these estimates on moments, the eigenvalue distributions are reconstructed. Lastly, HMC is applied to compute the moments and eigenvalue distributions for these models.

We find that all three methods have very close agreement and identified the critical coupling \eqref{critt2}, at which the $(1,0)$ model enjoys a phase transition from symmetric one-cut to asymmetric two-cut solutions. In addition to confirming the existence of a stable asymmetric regime for $(1,0)$ model, we have demonstrated the ability of the bootstrap method to correctly identify asymmetric solutions.

Future work will aim to apply the latter two methodologies to more complicated Dirac ensembles that correspond to multi-matrix models. Bootstrapping the type $(2,0)$, $(1,1)$, and $(0,2)$ quartic models was done in \cite{hessam2022bootstrapping}, however, only symmetric solutions were studied. We expect asymmetric solutions to these models' SDE to arise when not all odd moments are restricted to zero. Additionally, it would be interesting to apply these methods to Dirac ensembles with fermions, for which very little is known from an analytic perspective and nothing from a numerical perspective at the time of writing this article.

Finally, identification of the eigenvalue density of matrix $H$ is just the first step. The main appeal of the Dirac ensembles is the fact that they represent toy models of well-defined path integrals over geometries. It will therefore be desirable to understand the consequences of these asymmetric phase transitions for changes in the spectral geometry of the underlying fuzzy spaces.

\section*{Acknowledgements}
M. Khalkhali and N. Pagliaroli acknowledge the support of the Natural Sciences and Engineering Research Council of Canada (NSERC). The work of B. Bukor, S. Kov\'a\v cik and J. Tekel has been supported by project VEGA 1/0604/26 \emph{Quantum structures of spacetime}. We are grateful to Patrik Rusnák for his valuable help with the implementation and development of the bootstrap reconstruction method.

\section*{Data availability}
The code used to reproduce the figures and computations seen throughout the article will be shared upon request. 

\newpage

\appendix
\section{One-cut solutions}\label{AppA}
We present the proof of Theorem \ref{thm: one-cut density}, giving the form of all possible equilibrium measures with continuous density functions supported on a single interval.

\begin{proof}
Define the resolvent moment generating function as 
 $$W(x) = \int \frac{\rho(y)}{x-y}dy = \sum_{n=0}^{\infty}\frac{m_{n}}{x^{n+1}}.$$

 We can apply the Zhukovsky transform $$x(z) = \alpha + \gamma\left(z+\frac{1}{z}\right),$$ with $\alpha =\frac{a+b}{2}$ and $\gamma =\frac{b-a}{4}$. This change of variables allows us to write the resolvent as a Laurent polynomial in $z$: 
	\begin{equation*}
		W(x(z)) = \sum_{n=0}^{3}u_{n}z^{-n},
	\end{equation*}
	where
	\begin{align*}
		0 &=u_{0}= 4 \epsilon m_1 t_2 + 4 t_2 \alpha + 24 m_2 t_4 \alpha + 
 8 \epsilon m_3 t_4 \alpha + 24 \epsilon m_1 t_4 \alpha^2 + 8 t_4 \alpha^3 + 
 48 \epsilon m_1 t_4 \gamma^2 + 48 t_4 \alpha \gamma^2\\
		\frac{1}{\gamma}&= u_{1} = 4 t_2 \gamma + 24 m_2 t_4 \gamma + 8 \epsilon m_3 t_4 \gamma + 
 48 \epsilon m_1 t_4 \alpha \gamma + 24 t_4 \alpha^2 \gamma + 
 24 t_4 \gamma^3\\
		u_2 &= 24 \epsilon m_1 t_4 \gamma^2 + 24 t_4 \alpha \gamma^2\\
		u_{3} &= 8 t_{4}\gamma^3.
	\end{align*}
	The resolvent is equal to 
	\begin{equation*}
		W(x) =\frac{1}{2}\left(\frac{1}{2}Q(x)+ M(x)\sqrt{(x-a)(b-x)} \right),
	\end{equation*}
	where 
	\begin{align*}
		M(x) &= \frac{1}{\gamma}\sum_{n=1}^{3}u_{n}U_{n-1}\left(\frac{x-\alpha}{2 \gamma}\right)\\
 &= \frac{1}{\gamma^2} + 24 \epsilon m_{1} t_{4} (x-\alpha) + 8 t_{4} (x^2 + x \alpha - 2 \alpha^2 -\gamma^2),
	\end{align*}
	where $U_{n}(x)$ denotes the $n$-th Chebyshev polynomial of the second kind. For details on this transformation and simplification see Chapter 3 of \cite{eynard2016counting}.

	The moments can also be computed using the Zhukovsky transform via a change of variables
	\begin{align}\label{eq:residue formula for moments one-cut}
		m_{\ell} &= - \text{Res}_{x\rightarrow \infty} x^{\ell} W(x) dx\\
		&= - \text{Res}_{z\rightarrow \infty} \left(\alpha + \gamma\left(z+\frac{1}{z}\right) \right)^{\ell} W(x(z)) \left( \gamma\left(1-\frac{1}{z^2}\right)\right)dz\\
		&= \sum_{0\leq i,j\leq\ell} \frac{\ell!}{j! i!(\ell-i-j)!}\alpha^{\ell-i-j}\gamma^{i+j+1}(u_{i-j+1}-u_{i-j-1}).
	\end{align}
	
	Thus, we have:
	\begin{align*}
		m_{1} &=\alpha + 24 \epsilon m_1 t_4 \gamma^4 + 24 t_4 \alpha \gamma^4\\
		m_2 &=\gamma ^2+48 \alpha \gamma ^4 \epsilon m_{1} t_{4}+\alpha ^2 \left(48 \gamma ^4
 t_{4}+1\right)+8 \gamma ^6 t_{4}\\
		m_3 & =72 \alpha ^2 \gamma ^4 \epsilon m_{1} t_{4}+48 \gamma ^6 \epsilon m_{1}
 t_{4}+\alpha ^3 \left(72 \gamma ^4 t_{4}+1\right)+3 \alpha \left(\gamma ^2+24 \gamma ^6
 t_{4}\right).\\
	\end{align*}
	The first moment can then be solved for as
	\begin{equation}\label{eq:m1 one-cut}
		m_1 = \alpha \left(\frac{1+24 t_{4}\gamma^4}{1-24 t_{4}\epsilon\gamma^4}\right),
	\end{equation}
	subsequently allowing us to write
	\begin{align}\label{eq:m2 one-cut}
		m_{2} &=\gamma ^2-\frac{\alpha ^2 \left(24 \gamma ^4 (\epsilon+2) t_{4}+1\right)}{24 \gamma ^4
 \epsilon t_{4}-1}+8 \gamma ^6 t_{4}
 \end{align}
 \begin{align}\label{eq:m3 one-cut}
		m_3 &=\frac{3 \alpha \gamma ^2 \left(24 \gamma ^4 t_{4}+1\right) \left(8 \gamma ^4 \epsilon
 t_{4}-1\right)-\alpha ^3 \left(24 \gamma ^4 (2 \epsilon+3) t_{4}+1\right)}{24 \gamma ^4
 \epsilon t_{4}-1}
 \end{align}
 \begin{align}\label{eq:m4 one-cut}
 m_{4} &=\frac{6 \alpha ^2 \gamma ^2 \left(192 \gamma ^8 \epsilon t_{4}^2-8 \gamma ^4 (\epsilon+5)
 t_{4}-1\right)-\left(\alpha ^4 \left(24 \gamma ^4 (3 \epsilon+4) t_{4}+1\right)\right)+2
 \gamma ^4 \left(12 \gamma ^4 t_{4}+1\right) \left(24 \gamma ^4 \epsilon
 t_{4}-1\right)}{24 \gamma ^4 \epsilon t_{4}-1}.
	\end{align}
 
	We can now rewrite the two initial conditions coming from $u_0$ and $u_1$ as 
 \begin{align}\label{eq:1-cut condition 2}
 \begin{split}
		0&=\frac{4 \alpha}{1-24 \gamma ^4 \epsilon t_{4}} \left((\epsilon+1) t_{2}+2 t_{4} \left(3 \gamma ^2 \left(-2 \epsilon
 \left(96 \gamma ^8 t_{4}^2+12 \gamma ^4 t_{4}-1\right)+8 \gamma ^4
 t_{4}+3\right) \right.\right.\\
 &\left.\left.+\alpha ^3 \epsilon \left(24 \gamma ^4 (2 \epsilon+3)
 t_{4}+1\right)+\alpha ^2 \left(3 \epsilon \left(40 \gamma ^4 t_{4}+1\right)+144 \gamma
 ^4 t_{4}+4\right)-3 \alpha \gamma ^2 \epsilon \left(24 \gamma ^4 t_{4}+1\right) \left(8
 \gamma ^4 \epsilon t_{4}-1\right)\right)\right)
 \end{split}
	\end{align}
 and
	\begin{align}\label{eq:1-cut condition 1}
 \begin{split}
		1&= \frac{4 \gamma^2}{24 \gamma ^4 \epsilon
 t_{4}-1} \left(2 t_{4} \left(-\left(\alpha ^3 \left(24 \gamma ^4 \epsilon (2
 \epsilon+3) t_{4}+\epsilon\right)\right)-6 \alpha ^2 (\epsilon+1) \left(24 \gamma ^4
 t_{4}+1\right)\right. \right.\\
 &\left. \left.+3 \alpha \gamma ^2 \epsilon \left(24 \gamma ^4 t_{4}+1\right) \left(8
 \gamma ^4 \epsilon t_{4}-1\right)+6 \gamma ^2 \left(4 \gamma ^4 t_{4}+1\right) \left(24
 \gamma ^4 \epsilon t_{4}-1\right)\right)+t_{2}(24 \gamma ^4 \epsilon
 t_{4}-1)\right),
 \end{split}
	\end{align}
	respectively.
\end{proof}

\section{Two-cut solutions}\label{AppB}
Consider the continuous density function for the equilibrium measure with support $[a_{1},b_{1}]\cup[a_{2},b_{2}]$:
\begin{equation*}
 \rho(x) = \frac{1}{2 \pi} \left(24 \epsilon t_4 m_{1} +8 t_4 x + 4t_4 (a_{1}+a_{2}+b_{1}+b_{2})\right)\text{cut}(x)\sqrt{-(x-a_{1})(x-a_{2})(x-b_{1})(x-b_{2})}_{[a_{1},b_{1}]\cup[a_{2},b_{2}]},
\end{equation*}
where 
\begin{equation*}
 \operatorname{cut}(x)= \begin{cases}-1 & x \in\left[a_1, b_1\right] \\ +1 & x \in\left[a_2, b_2\right]\end{cases}.
\end{equation*}
From the discussion in Section \ref{sec:RH} we know that we have the following constraints:
\begin{align}\label{eq:2-cut C1}
 0 &= \int_{\text{supp}(\rho)} \frac{Q(s)}{(\sqrt{q(z)})_{+}}ds
 \end{align}
 \begin{align}\label{eq:2-cut C2}
 0 &= \int_{\text{supp}(\rho)} \frac{Q(s)}{(\sqrt{q(z)})_{+}}sds
 \end{align}
 \begin{align}\label{eq:2-cut C3}
 1 &= \frac{i}{2 \pi}\int_{\text{supp}(\rho)} \frac{Q(s)}{(\sqrt{q(z)})_{+}}s^{2}ds
 \end{align}
 \begin{align}\label{eq:2-cut C4}
 0 &= \int_{\text{supp}(\rho)} \ln\frac{|a_{2}-y|}{|b_{1}-y|}\rho(y) dy-\frac{1}{2}
 \int_{b_1}^{a_2} Q(x) dx 
 \end{align}
 \begin{align}\label{eq:2-cut C5}
 m_{1} &= \int_{\text{supp}(\rho)} x \rho(x) dx
 \end{align}
 \begin{align}\label{eq:2-cut C6}
 m_{2} &= \int_{\text{supp}(\rho)} x^2 \rho(x) dx
 \end{align}
 \begin{align}\label{eq:2-cut C7}
 m_{3} &= \int_{\text{supp}(\rho)} x^3 \rho(x) dx.
\end{align}
We note that in \cite{d2026symmetry} a different off-support condition was found using Lagrange multipliers instead of condition \eqref{eq:2-cut C4}.

Using contour integration, we can write 
\begin{equation*}
 \int_{\text{supp}(\rho)} \frac{\frac{i}{\pi}Q(s)}{(\sqrt{q(z)})_{+}}s^{\alpha}ds = -\pi i\res_{s =0}\frac{Q\left(\frac{1}{s}\right)}{s^{2+\alpha} \sqrt{q\left(\frac{1}{s}\right)}}. 
\end{equation*}
The Laurent expansion is computed as 
\begin{equation*}
 \frac{Q\left(\frac{1}{s}\right)}{s^{2+\alpha} \sqrt{q\left(\frac{1}{s}\right)}} = \sum_{i_{1},i_{2},j_{1},j_{2} =0}^\infty {\frac{1}{2} \choose i_{1}} {\frac{1}{2} \choose i_{2}} {\frac{1}{2} \choose j_{1}} {\frac{1}{2} \choose j_{2}}(-1)^{i_{1} +i_{2}+j_{1}+j_{2}}a_{1}^{i_{1}}a_{2}^{i_{2}}b_{1}^{j_{1}}b_{2}^{j_{2}} s^{-\alpha+i_{1} +i_{2}+j_{1}+j_{2}}Q\left(\frac{1}{s}\right),
\end{equation*}
where 
\begin{equation*}
 Q\left(\frac{1}{s}\right) = \frac{4 t_{4}}{s^3}+ 12 \frac{\epsilon t_4 m_{1}}{s^2} + \frac{2 t_{2} + 12 t_{4}m_{2}}{s} + \epsilon( 2t_{2} m_{1} + 4 t_{4}m_{3}).
\end{equation*}
The first three conditions \eqref{eq:2-cut C1}, \eqref{eq:2-cut C2}, and \eqref{eq:2-cut C3} can be simplified to
\begin{align*}
\begin{split}
 0&=6 (a_1 + a_2 + b_1 + b_2) \epsilon m_1 + 2 t_2 \\
 &+ 
 \frac{t_4}{2} (3 a_1^2 + 2 a_1 a_2 + 3 a_2^2 + 2 a_1 b_1 + 2 a_2 b1 + 3 b_1^2 + 
 2 a_1 b_2 + 2 a_2 b_2 + 2 b_1 b_2 + 3 b_2^2) + 12 m_2 t_4,
 \end{split}
\end{align*}

\begin{align*}
 \begin{split}
 0&=\frac{1}{4} \big(6 \epsilon m_{1} \left(3 a_{1}^2+2 a_{1}
 (a_{2}+b_{1}+b_{2})+3 a_{2}^2+2 a_{2} (b_{1}+b_{2})+3 b_{1}^2+2
 b_{1} b_{2}+3 b_{2}^2\right)\\
 &+t_{4} \left(5 a_{1}^3+3 a_{1}^2
 (a_{2}+b_{1}+b_{2})+a_{1} \left(3 a_{2}^2+2 a_{2} (b_{1}+b_{2})+3
 b_{1}^2+2 b_{1} b_{2}+3 b_{2}^2\right)\right.\\
 &\left.+5 a_{2}^3+3 a_{2}^2
 (b_{1}+b_{2})+a_{2} \left(3 b_{1}^2+2 b_{1} b_{2}+3 b_{2}^2\right)+5
 b_{1}^3+3 b_{1}^2 b_{2}+3 b_{1} b_{2}^2+5 b_{2}^3\right)\\
 &+4 (6 m_{2}
 t_{4}+t_{2}) (a_{1}+a_{2}+b_{1}+b_{2})+8 \epsilon (m_{1} t_{2}+2
 m_{3} t_{4})\big),
 \end{split}
\end{align*}
and
\begin{align*}
 \begin{split}
 1&= 8 (6 m_{2} t_{4}+t_{2}) \left(3 a_{1}^2+2 a_{1} (a_{2}+b_{1}+b_{2})+3
 a_{2}^2+2 a_{2} (b_{1}+b_{2})+3 b_{1}^2+2 b_{1} b_{2}+3
 b_{2}^2\right)\\
 &+24 \epsilon m_{1} (5 a_{1}^3+3 a_{1}^2
 (a_{2}+b_{1}+b_{2})+a_{1} \left(3 a_{2}^2+2 a_{2} (b_{1}+b_{2})+3
 b_{1}^2+2 b_{1} b_{2}+3 b_{2}^2\right)\\
 &+5 a_{2}^3+3 a_{2}^2
 (b_{1}+b_{2})+a_{2} \left(3 b_{1}^2+2 b_{1} b_{2}+3 b_{2}^2\right)+5
 b_{1}^3+3 b_{1}^2 b_{2}+3 b_{1} b_{2}^2+5 b_{2}^3)\\
 &+t_{4} \big(35
 a_{1}^4+20 a_{1}^3 (a_{2}+b_{1}+b_{2})+6 a_{1}^2 \left(3 a_{2}^2+2
 a_{2} (b_{1}+b_{2})+3 b_{1}^2+2 b_{1} b_{2}+3 b_{2}^2\right)\\
&+4 a_{1}
 \left(5 a_{2}^3+3 a_{2}^2 (b_{1}+b_{2})+a_{2} \left(3 b_{1}^2+2 b_{1}
 b_{2}+3 b_{2}^2\right)+5 b_{1}^3+3 b_{1}^2 b_{2}+3 b_{1} b_{2}^2+5
 b_{2}^3\right)\\
 &+35 a_{2}^4+20 a_{2}^3 (b_{1}+b_{2})+6 a_{2}^2 \left(3
 b_{1}^2+2 b_{1} b_{2}+3 b_{2}^2\right)+4 a_{2} \left(5 b_{1}^3+3 b_{1}^2
 b_{2}+3 b_{1} b_{2}^2+5 b_{2}^3\right)\\
 &+35 b_{1}^4+20 b_{1}^3 b_{2}+18
 b_{1}^2 b_{2}^2+20 b_{1} b_{2}^3+35 b_{2}^4\big)+32 \epsilon 
 (a_{1}+a_{2}+b_{1}+b_{2}) (m_{1} t_{2}+2 m_{3} t_{4}).
 \end{split}
\end{align*}

The last three conditions \eqref{eq:2-cut C5}, \eqref{eq:2-cut C6}, and \eqref{eq:2-cut C7} can be simplified by computing moments
\begin{align*}
 \int x^{n} \rho(x) dx &= -\frac{1}{2 \pi i}\oint z^{n} W(z)dz\\
 &=-\res_{z \rightarrow \infty}[z^{-n-2} \rho(z^{-1})],
\end{align*}
where the contour is a large circle that encompasses the support of $\rho$. By expanding 
\begin{equation*}
 z^{-n-5}\sqrt{(1-a_{1}z)(1-a_{2}z)(1-b_{1}z)(1-b_{2}z)} = \sum_{i_{1},i_{2},i_{3},i_{4}=1} {\frac{1}{2}\choose i_{1}}{\frac{1}{2}\choose i_{2}}{\frac{1}{2}\choose i_{3}}{\frac{1}{2}\choose i_{4}} (-z)^{i_{1}+i_{2}+i_{3}+i_{4} -n-5},
\end{equation*}
these residues can be computed.
For example,
\begin{equation}
 \begin{split}
 m_{1}&= -\frac{1}{32} (6 a_{1}^5 - 2 a_{1}^4 a_{2} - 4 a_{1}^3 a_{2}^2 - 4 a_{1}^2 a_{2}^3 - 
 2 a_{1} a_{2}^4 + 6 a_{2}^5 - 2 a_{1}^4 b_{1} + 4 a_{1}^2 a_{2}^2 b_{1} - 2 a_{2}^4 b_{1} - 
 4 a_{1}^3 b_{1}^2 \\
 &+ 4 a_{1}^2 a_{2} b_{1}^2 + 4 a_{1} a_{2}^2 b_{1}^2 - 4 a_{2}^3 b_{1}^2
 - 
 4 a_{1}^2 b_{1}^3 - 4 a_{2}^2 b_{1}^3 - 2 a_{1} b_{1}^4 - 2 a_{2} b_{1}^4 + 6 b_{1}^5
 - 
 2 a_{1}^4 b_{2} + 4 a_{1}^2 a_{2}^2 b_{2} - 2 a_{2}^4 b_{2} + 4 a_{1}^2 b_{1}^2 b_{2}\\
 &+ 
 4 a_{2}^2 b_{1}^2 b_{2} - 2 b_{1}^4 b_{2} 
 - 4 a_{1}^3 b_{2}^2 + 4 a_{1}^2 a_{2} b_{2}^2
 + 
 4 a_{1} a_{2}^2 b_{2}^2 - 4 a_{2}^3 b_{2}^2 + 4 a_{1}^2 b_{1} b_{2}^2 + 4 a_{2}^2 b_{1} b_{2}^2 + 
 4 a_{1} b_{1}^2 b_{2}^2 + 4 a_{2} b_{1}^2 b_{2}^2 - 4 b_{1}^3 b_{2}^2 - 4 a_{1}^2 b_{2}^3\\
 &- 
 4 a_{2}^2 b_{2}^3 - 4 b_{1}^2 b_{2}^3 - 2 a_{1} b_{2}^4 - 2 a_{2} b_{2}^4 - 2 b_{1} b_{2}^4 + 
 6 b_{2}^5 + 15 a_{1}^4 \epsilon m_{1} - 12 a_{1}^3 a_{2} \epsilon m_{1} - 6 a_{1}^2 a_{2}^2 \epsilon m_{1} \\
 &- 
 12 a_{1} a_{2}^3 \epsilon m_{1} + 15 a_{2}^4 \epsilon m_{1} - 12 a_{1}^3 b_{1} \epsilon m_{1} 
 + 
 12 a_{1}^2 a_{2} b_{1} \epsilon m_{1} + 12 a_{1} a_{2}^2 b_{1} \epsilon m_{1} - 12 a_{2}^3 b_{1} \epsilon m_{1} - 
 6 a_{1}^2 b_{1}^2 \epsilon m_{1}\\
 &+ 12 a_{1} a_{2} b_{1}^2 \epsilon m_{1} - 6 a_{2}^2 b_{1}^2 \epsilon m_{1}- 
 12 a_{1} b_{1}^3 \epsilon m_{1} - 12 a_{2} b_{1}^3 \epsilon m_{1} + 15 b_{1}^4 \epsilon m_{1} - 
 12 a_{1}^3 b_{2} \epsilon m_{1}\\
 &+ 12 a_{1}^2 a_{2} b_{2} \epsilon m_{1} + 12 a_{1} a_{2}^2 b_{2} \epsilon m_{1} - 
 12 a_{2}^3 b_{2} \epsilon m_{1} + 12 a_{1}^2 b_{1} b_{2} \epsilon m_{1} - 24 a_{1} a_{2} b_{1} b_{2} \epsilon m_{1} + 
 12 a_{2}^2 b_{1} b_{2} \epsilon m_{1}\\
 &+ 12 a_{1} b_{1}^2 b_{2} \epsilon m_{1} + 12 a_{2} b_{1}^2 b_{2} \epsilon m_{1}- 
 12 b_{1}^3 b_{2} \epsilon m_{1} - 6 a_{1}^2 b_{2}^2 \epsilon m_{1} + 12 a_{1} a_{2} b_{2}^2 \epsilon m_{1}\\
 &- 
 6 a_{2}^2 b_{2}^2 \epsilon m_{1} + 12 a_{1} b_{1} b_{2}^2 \epsilon m_{1} + 12 a_{2} b_{1} b_{2}^2 \epsilon m_{1} - 
 6 b_{1}^2 b_{2}^2 \epsilon m_{1} - 12 a_{1} b_{2}^3 \epsilon m_{1} - 12 a_{2} b_{2}^3 \epsilon m_{1} - 
 12 b_{1} b_{2}^3 \epsilon m_{1} + 15 b_{2}^4 \epsilon m_{1}).
 \end{split}
\end{equation}
This first moment condition allows us to isolate for $m_{1}$ strictly in terms of the limit points
\begin{equation}\label{eq:2-cut m1}
 \resizebox{1.05\hsize}{!}{$m_{1}=-\frac{2 \left(-3 a_{1}^5+a_{1}^4 (a_{2}+b_{1}+b_{2})+2 a_{1}^3
 \left(a_{2}^2+b_{1}^2+b_{2}^2\right)+2 a_{1}^2 \left(a_{2}^3-a_{2}^2
 (b_{1}+b_{2})-a_{2} \left(b_{1}^2+b_{2}^2\right)+(b_{1}-b_{2})^2
 (b_{1}+b_{2})\right)+a_{1} \left(a_{2}^4-2 a_{2}^2
 \left(b_{1}^2+b_{2}^2\right)+\left(b_{1}^2-b_{2}^2\right)^2\right)-3
 a_{2}^5+a_{2}^4 (b_{1}+b_{2})+2 a_{2}^3 \left(b_{1}^2+b_{2}^2\right)+2
 a_{2}^2 (b_{1}-b_{2})^2 (b_{1}+b_{2})+a_{2}
 \left(b_{1}^2-b_{2}^2\right)^2-(b_{1}-b_{2})^2 \left(3 b_{1}^3+5 b_{1}^2
 b_{2}+5 b_{1} b_{2}^2+3 b_{2}^3\right)\right)}{-15 a_{1}^4+12 a_{1}^3
 (a_{2}+b_{1}+b_{2})+6 a_{1}^2 \left(a_{2}^2-2 a_{2}
 (b_{1}+b_{2})+(b_{1}-b_{2})^2\right)+12 a_{1} \left(a_{2}^3-a_{2}^2
 (b_{1}+b_{2})-a_{2} (b_{1}-b_{2})^2+(b_{1}-b_{2})^2
 (b_{1}+b_{2})\right)-15 a_{2}^4+12 a_{2}^3 (b_{1}+b_{2})+6 a_{2}^2
 (b_{1}-b_{2})^2+12 a_{2} (b_{1}-b_{2})^2 (b_{1}+b_{2})-15
 b_{1}^4+12 b_{1}^3 b_{2}+6 b_{1}^2 b_{2}^2+12 b_{1} b_{2}^3-15
 b_{2}^4+32}.$}
\end{equation}

The last two conditions are in terms of the limit points and $m_{1}$, hence, we can write all three moments in terms of the limit points in this manner. With all moments in terms of the limit points, the equations \eqref{eq:2-cut C1},\eqref{eq:2-cut C2},
\eqref{eq:2-cut C3}, and \eqref{eq:2-cut C4} give us four equations with four unknowns. Explicit solutions of this system can be found, but, with the exception of the symmetric solution, are extremely cumbersome to work with. For example, the first 
condition \eqref{eq:2-cut C1} gives us $m_{2}$ in terms of the limit points and $m_{1}$:
\begin{align}
\begin{split}
 m_{2}&= \frac{1}{24 t_4}(-12 a_1 m_1 - 12 a_2 m_1 - 12 b_1 m_1 - 12 b_2 m_1 - 4 t_2 - 
 3 a_1^2 t_4 \\
 &- 2 a_1 a_2 t_4 - 3 a_2^2 t_4 - 2 a_1 b_1 t_4 - 2 a_2 b_1 t_4 - 
 3 b_1^2 t_4 - 2 a_1 b_2 t_4 - 2 a_2 b_2 t_4 - 2 b_1 b_2 t_4 - 3 b_2^2 t_4),
 \end{split}
\end{align}
where $m_{1}$ is as in \eqref{eq:2-cut m1}. 

Assuming $\text{supp}(\rho) = [-b,-a]\cup [a,b]$, conditions \eqref{eq:2-cut C1}, \eqref{eq:2-cut C3}, and \eqref{eq:2-cut C6} become 
\begin{align*}
 4 t_{2} + 24 m_{2} t_{4} + 4 t_{4} (a^2 + b^2) =0,
\end{align*}
\begin{equation*}
 3 t_{4}(a^4 + b^4 ) + 2 b^2( t_{2} + 6m_{2} t_{4}) + 2 a^2 ( t_{2} + (b^2 + 6 m_{2})t_{4}) =2,
\end{equation*}
and 
\begin{equation*}
 m_{2} =-\frac{1}{8 t_{4}}\left(4 t_{2} + 24 m_{2} t_{4}\right).
\end{equation*}
The last condition can be solved to arrive at 
$$m_{2} =-\frac{t_{2}}{8t_{4}}.$$

The first two conditions here can then be solved to show
\begin{align*}
\begin{split}
\rho(x) = \frac{4t_{4}}{\pi}|x| \sqrt{(x^{2}-a^{2})(b^{2}-x^{2})}_{[-b,-a]\cup [b,a]},
\end{split}
\end{align*}
where the support $[-b,-a]\cup [b,a]$ is
\begin{equation*}
a^{2}= \frac{-t_{2}+ 4 \sqrt{2t_{4}}}{8t_{4}},
\end{equation*}
and 
\begin{equation*}
b^{2}= -\frac{t_{2}+ 4 \sqrt{2t_{4}}}{8t_{4}},
\end{equation*}
for $t_{4}\geq 0$.
\section{The one-cut free energy computation}\label{appC}
Here we present the computationally heavy proof of Corollary \ref{cor: sym free energy}.
\begin{proof}
    The integral dramatically simplifies if we use the saddle point equation to write it in terms of the entropy of $\rho$. Consider the free energy functional
\begin{equation*}
 F = \int \int(U(x,y) - \ln|x-y|)\rho(x)\rho(y)dx dy. 
\end{equation*}
Integrating the saddle point equation, we have that 
\begin{equation*}
 \int_{\text{supp}( \rho)}( \ln|x-y|) \rho(y)dy = \int_{\text{supp}( \rho)} \ln|s| \rho(s)ds +\frac{1}{2}\int_{\text{supp}( \rho)} U(x,y)\rho(y)dy,
\end{equation*}
which can be combined to give us
\begin{equation*}
 F = -\int_{\text{supp}( \rho)} \ln|s| \rho(s)ds + \frac{1}{2}\int_{\text{supp}( \rho)}\int_{\text{supp}( \rho)} U(x,y)\rho(x)\rho(y)dx dy.
\end{equation*}
The latter integral can be computed in terms of the moments of $\rho$ since $U(x,y)$ is a polynomial. 

For $\ell\geq 0$, define
$$L_{\ell} :=\int_{-2\gamma}^{2\gamma} \ln|x|x^{\ell} \sqrt{4\gamma^2-x}.$$
With the substitution $x \rightarrow 2 \gamma \cos \theta$, this integral becomes
\begin{equation*}
    L_{\ell} =2 \gamma \int_{0}^{\pi} \ln |2 \gamma \cos \theta| (2 \gamma \cos \theta)^{\ell}(1 - \cos2 \theta) d\theta.
\end{equation*}
We may Fourier series expand the integrand,
\begin{equation*}
    \ln |2 \gamma \cos \theta| = \ln \gamma + \sum_{n=1}^{\infty} \frac{(-1)^{n+1}}{n} \cos (2 n \theta)
\end{equation*}
which has a finite even Fourier expansion, so we substitute it back into the integral. We find then that 
\begin{align*}
    L_{0} &= \pi \gamma^2 (2 \ln |\gamma| -1),
    \end{align*}
and
    \begin{align*}
    L_{2}&= \pi \gamma^2 \left(2 \ln \gamma + \frac{1}{2}\right).
\end{align*}
Then applying the formula for $\rho(x)$ from Theorem \ref{thm: one-cut density}, we have that 
\begin{align*}
\int_{-2\gamma}^{2 \gamma}\ln |x| \rho(x)dx&= \frac{1}{2\pi}\left(8 t_{4} L_{2} + \left( \frac{1}{\gamma^2} - 8 t_{4} \gamma^2\right)L_{0} \right)\\
&=\ln\gamma + 6 t_{4} \gamma^{4} -\frac{1}{2}. 
\end{align*}
Lastly, to compute the integral of $U(x,y)\rho(x)\rho(y)$ we can use equations \eqref{eq:m2 one-cut} for the second moment and compute $m_{4}$ in a similar manner using \eqref{eq:residue formula for moments one-cut}.

For the symmetric 2-cut case, the same approach works but one requires the substitution $a^2 = s-t$ and $b^2 = s+t$, and take $x \rightarrow s+ t \cos \theta$. Almost the exact transformation is used in the next proof.

\end{proof}

We have an included the analogous computation for an asymmetric case.
\begin{prop}\label{free energy theorem}
 	The free energy for a one-cut solution of the type $(1,0)$ quartic Dirac ensemble, assuming $0<a<b$ is of the form 
 \begin{equation}
 I[\rho_{\textit{1-cut}}] = E(t_{2},t_{4}) +2t_{4}(m_{4} + 3 m_{2}^2 + 4 \epsilon m_{1}m_{3}) + 2t_{2}(m_{2}+ \epsilon m_{1}^2),
 \end{equation}
 where
 \begin{itemize}
 \item the first term is \begin{align*}&E(t_{2},t_{4})=
\gamma^2\left( 8 t_{4}\alpha^2+B\alpha+C\right)
\left(
\ln\left(\frac{q}{2}\right)+\frac{2\gamma^2}{q^2}
\right)
+
\frac{2\gamma^4(16 t_{4}\alpha+B)}{q}
\left(
1-\frac{4\gamma^2}{3q^2}
\right)
+
8 t_{4}\gamma^4
\left(
\ln\left(\frac{q}{2}\right)+\frac{4\gamma^4}{q^4}
\right),
 \end{align*}
 where $q = \alpha + \sqrt{\alpha^2 - 4 \gamma^2}$,
 \begin{align*}
 B& = 8 \alpha t_{4}-\frac{24 \epsilon t_{4} \left(\alpha +24 \alpha \gamma ^4 t_{4}\right)}{24 \gamma ^4 \epsilon
 t_{4}-1},
 \end{align*}
 and
 \begin{align*}
 C & = \frac{1}{\gamma ^2}+\frac{24 \alpha \epsilon t_{4} \left(\alpha +24 \alpha \gamma ^4 t_{4}\right)}{24 \gamma ^4
 \epsilon t_{4}-1}-16 \alpha ^2 t_{4}-8 \gamma ^2 t_{4}.
 \end{align*}
 \item $m_{1},m_{2},$ $m_{3},$ and $m_{4}$ are as in (\ref{eq:m1 one-cut}-\ref{eq:m4 one-cut}) 
 \item $\gamma = \frac{b-a}{4}$ and $\alpha=\frac{a+b}{2}$, and are the solutions of equations \eqref{eq:1-cut condition 1} and \eqref{eq:1-cut condition 2}.
 \end{itemize}
\end{prop}

\begin{proof}
As in the previous proof, we wish to compute integrals of the form 
\begin{equation*}
 L_{\ell}:= \int_{a}^{b} \ln|x| x^\ell \sqrt{(x-a)(b-x)}dx.
\end{equation*}
for $\ell = 0,1,$ and 2.

Applying the substitution $a=s-t$ and $b= s+t$, we can write the discriminant as $t^2 - (x-s)^2$. Then making the change of variables $x \rightarrow s+t\cos (\theta)$ we have that
\begin{align*}
 L_{\ell}
 &=\frac{t^2}{2}\int_{0}^{\pi} \ln|s+t\cos (\theta)| (s+t\cos (\theta))^\ell \left(1-\cos(2 \theta)\right) d \theta.
\end{align*}
We may write the integrand as a Fourier cosine series 
\begin{equation*}
 \ln|s+ t \cos(\theta)| =  \sum_{n=0}^{\infty} a_{n}\cos( n \theta).
\end{equation*}

To compute the Fourier coefficients we start by mapping the interval $[0,\pi]$ to the unit circle with $z = e^{i\theta}$. Let $\alpha = \frac{-s+\sqrt{s^2-t^2}}{t}$ and $\beta = \frac{-s-\sqrt{s^2-t^2}}{t}= \frac{1}{\alpha}$ denote the roots of the argument of the logarithm after this transformation. 

For the case that $\ell=0$, we have that 
$$
 \frac{t^2}{2}\int_0^\pi\left(\frac{a_0}{2}+\sum_{n \geq 1} a_n \cos (n \theta)\right)\left(1-\cos (2 \theta)\right) d \theta= \frac{\pi t^2}{4}\left(a_0-a_2\right).
$$
Then we may use the Residue Theorem to compute
\begin{align*}
 a_{0} &= \frac{1}{\pi}\int_{0}^{\pi}\ln|s+t \cos(\theta)|d\theta \\
 &=\frac{1}{\pi i} \oint_{|z|=1}\ln\left| s+ \frac{t}{2}\left(z+\frac{1}{z} \right)\right|\frac{dz}{z}\\
 &=2\text{Res}_{z\rightarrow 0}\left[-\frac{1}{z^2}\ln\left( s+ \frac{t}{2}\left(z+\frac{1}{z} \right)\right)\right]\\
 &= 2 \text{Res}_{z\rightarrow 0}\left[\frac{1}{z}\ln \left(\frac{s+\sqrt{s^2-t^2}}{2}\right)-\frac{1}{z}\left(\sum_{n=1}^{\infty} \frac{\alpha^n}{n z^n}+\sum_{n=1}^{\infty} \frac{\alpha^n z^n}{n} \right)\right]\\
 &= 2\ln \left(\frac{s+\sqrt{s^2-t^2}}{2}\right),
\end{align*}
and via the same approach
\begin{align*}
 a_{n} = \frac{2(-1)^{n+1}}{n} r^{n},
\end{align*}
where $r = \frac{t}{s + \sqrt{s^2 -t^2}}$, for $n\geq 1$. Thus
\begin{equation*}
 L_{0} = \frac{\pi t^2}{4} \left( 2\ln \left(\frac{s+\sqrt{s^2-t^2}}{2}\right) +r^2 \right).
\end{equation*}

When $\ell =1$, we have 
\begin{align*}
 L_{1} &=\frac{st^2}{2}\int_{0}^{\pi} \ln|s+t\cos (\theta)| \left(1-\cos(2 \theta)\right) d \theta
 + \frac{t^3}{2}\int_{0}^{\pi} \ln|s+t\cos (\theta)| \cos (\theta)\left(1-\cos(2 \theta)\right) d \theta.
\end{align*}
The Fourier cosine series expansion of the second term is 
$$
 \frac{t^2}{2}\int_0^\pi\left(\frac{a_0}{2}+\sum_{n \geq 1} a_n \cos ( n \theta)\right)\cos(\theta)\left(1-\cos (2\theta)\right) d \theta= \frac{\pi t^2 }{8}(a_{1} -a_{3}).
$$
Hence, 
\begin{equation*}
 L_{1} = s L_{0} + \frac{\pi t^{3}}{4}\left(r - \frac{r^3}{3} \right).
\end{equation*}

When $\ell =2$, we have 
\begin{align*}
 L_{2} &=\frac{t^2s^2}{2}\int_{0}^{\pi} \ln|s+t\cos (\theta)| \left(1-\cos(2 \theta)\right) d \theta +st^3\int_{0}^{\pi} \ln|s+t\cos (\theta)|\cos(\theta) \left(1-\cos(2 \theta)\right) d \theta\\
 &+\frac{t^4}{2}\int_{0}^{\pi} \ln|s+t\cos (\theta)| \cos(\theta)^2\left(1-\cos(2 \theta)\right) d \theta.\\
\end{align*}
The Fourier cosine series for the third integral is 
\begin{equation*}
 \frac{t^{4}}{2}\int_0^\pi\left(\frac{a_0}{2}+\sum_{n \geq 1} a_n \cos (n \theta)\right)\cos(\theta)^2\left(1-\cos (2 \theta)\right) d \theta= \frac{\pi t^{4}}{16}(a_{0}-a_{4}).
\end{equation*}
Hence, 
\begin{equation*}
 L_{2} = s^2 L_{0} + \frac{\pi s t^3}{2}\left( r - \frac{r^3}{3}\right) + \frac{\pi t^{4}}{8}\left(\ln \left(\frac{s+\sqrt{s^2-t^2}}{2}\right) - \frac{r^{4}}{4}\right).
\end{equation*}

For a density 
$\rho(x) = \frac{1}{2\pi}(Ax^2 + Bx + C)\sqrt{(x-a)(b-x)}$, we can compute 
\begin{align*}
 \int_{a}^{b} \ln|x| \rho(x) dx &= \frac{1}{2 \pi}( AL_{2} + B L_{1} + CL_{0}).
\end{align*}
Applying the inverse change of variables $s =\frac{a+b}{2}$ and $t =\frac{b-a}{2}$, and using the appropriate choice of $A$, $B$, and $C$ from Theorem \ref{thm: one-cut density} gives us the formula for $E(t_{2},t_{4})$.

Lastly, to compute the integral of $U(x,y)\rho(x)\rho(y)$ we can use equations \eqref{eq:m1 one-cut},\eqref{eq:m2 one-cut},\eqref{eq:m3 one-cut} for the moments $m_{1},m_{2}$ and $m_{3}$, but we also require $m_{4}$, which can be computed in the same manner using the residue formula \eqref{eq:residue formula for moments one-cut}.
\end{proof}

A similar formula can be computed in the case that $a<0<b$ or $a<b<0$, with slightly different substitutions.
\bibliographystyle{unsrt}
\bibliography{references}
\end{document}